%% file: main-2.tex
\documentclass[11pt]{article}

\input{0-macros}

\title{Graph Spectral Sparsification is in Catalytic Logspace}
\author{Cassandra Marcussen\\Harvard University\\\texttt{cmarcussen@g.harvard.edu} \and Edward Pyne\\MIT\\\texttt{epyne@mit.edu} \and Ronitt Rubinfeld\\MIT\\\texttt{ronitt@csail.mit.edu}}
\begin{document}
\maketitle
\begin{abstract}
    We give a catalytic logspace algorithm for the problem of graph spectral sparsification. Given an undirected graph $G$ on $n$ vertices and $\eps>0$, our algorithm outputs an $\eps$-spectral sparsifier of $G$ with $O(n\eps^{-2}\log n)$ edges, matching the effective resistance sampling of Spielman and Srivastava (STOC 2008). This gives a new, natural problem in catalytic logspace that is not known to be in deterministic $\NC$ or $\SC$. Our main contribution is an entirely new technique in the compress--or--random paradigm for catalytic logspace that we believe will have further applications. 
    
    We first analyze effective-resistance sparsification using a pessimistic estimator that can itself be computed in catalytic logspace. The estimator is motivated by the viewpoint of graph quasirandomness and immediately gives a simple, deterministic greedy algorithm for graph sparsification.
    
    Subsequently, we show that such a pessimistic estimator can be transformed into an algorithm that performs an in-place compression of a string with bad potential. Our algorithm is based on using the potential function to define a measure over strings, and implementing arithmetic coding using this measure in-place.  This compression technique is substantially distinct from all prior tools in the field of catalytic computation.
\end{abstract}

\section{Introduction}
Graph sparsification has proven to be a fundamental primitive across algorithms and complexity (see the survey of Batson, Spielman, Srivastava, and Teng \cite{DBLP:journals/cacm/BatsonSST13}). Initially defined as cut sparsifiers \cite{DBLP:conf/stoc/BenczurK96}, the work of Spielman and Teng  \cite{DBLP:journals/siamcomp/SpielmanT11} generalized cut sparsification to spectral sparsification:
\begin{definition}[Spectral Approximation]
    For an (undirected) graph $G$ with Laplacian $L_G$, we say a graph $H$ is an $\ep$-spectral approximation of $G$ if 
    \[
    (1-\ep)L_G\preccurlyeq L_H \preccurlyeq (1+\ep)L_G
    \]
    where $\preccurlyeq$ denotes the Loewner order on PSD matrices.
\end{definition}
Spielman and Teng showed that every graph $G$ has an $\ep$-spectral sparsifier with $\tO(n/\ep^2)$ edges, and gave a randomized algorithm to construct such a sparsifier.
While much existing work~\cite{SS11, DBLP:journals/siamcomp/BatsonSS12, DBLP:conf/stoc/Lee017} has the goal of constructing sparsifiers in small \textit{time}, several recent works have explored the goal of constructing sparsifiers in small \textit{space}~\cite{DBLP:conf/icalp/AhnG09, DBLP:conf/pods/AhnGM12, DBLP:journals/mst/KelnerL13, DBLP:journals/siamcomp/KapralovLMMS17, DBLP:journals/toct/DoronMVZ24}. 

Simultaneously, a new model of space-bounded computation known as catalytic computing has seen an explosion of recent work \cite{DBLP:conf/stoc/BuhrmanCKLS14, DBLP:journals/eatcs/Koucky16, DBLP:journals/mst/BuhrmanKLS18, DBLP:journals/toct/DattaGJST26, DBLP:conf/stoc/CookM24, DBLP:conf/stoc/CookLMP25,DBLP:journals/eatcs/Mertz23, DBLP:journals/cc/Pyne25, DBLP:journals/jacm/Williams26}. 
\begin{definition}[Catalytic Logspace]
    We say a language $L$ is in catalytic logspace ($\CL$) if there is a machine $A$ that works as follows on every input $x$ and catalytic tape content $\tau$. The machine has a read-only tape holding $x$, a worktape of length $O(\log n)$, and a catalytic tape of length $\zo^{n^c}$ holding $\tau$. At the end of the computation, the machine returns $L_x$ and the catalytic tape reads $\tau$. 
\end{definition}
Catalytic algorithms seem to have powers beyond that of standard logspace. For instance, they can compute matrix determinants (complete for the class $\DET$) and $s\ra t$ connectivity in directed graphs (complete for the class $\NL$). Recent work has built practical $\CL$ algorithms for interesting problems~\cite{DBLP:conf/stoc/CookM24}, and shown that the fundamental problem of graph matching lies in $\CL$ \cite{DBLP:conf/focs/AgarwalaM25,DBLP:conf/innovations/AgarwalaAV26,CDKMS26}. However, there was no known way to use the powers of catalytic computing to construct even the weaker notion of cut sparsifiers. 

We show that catalytic logspace algorithms can produce graph spectral sparsifiers of near-linear size. We first recall the definition of search problems in catalytic logspace:
\begin{definition}[Search catalytic logspace]
    For a relation $R\subseteq \zo^*\times \zo^*$, we say that $R$ is in $\searchCL$ if there is a valid catalytic logspace machine that on input $x$ and starting tape $\tau$, outputs $y_{x,\tau}$ such that $(x,y_{x,\tau})\in R$.\footnote{Note that we allow the output of the algorithm to depend on the initial catalytic tape (and our algorithm, like every other compress--or--random algorithm, will exploit this freedom), but the output must always be valid according to the relation.}
\end{definition}

\begin{theorem}\label{thm:main}
    There is a $\searchCL$ algorithm that, given an undirected graph $G$ on $n$ vertices and $\eps>0$, outputs a (weighted) graph $H$ with $O(n\eps^{-2}\log n)$ edges such that $H$ is an $\eps$-spectral approximation of $G$.
\end{theorem}

We first compare our work to what was previously known. The most relevant works are the leverage-score-based sparsifier of Spielman and Srivastava~\cite{SS11} and the (partial) derandomization of their work by Doron, Murtagh, Vadhan, and Zuckerman~\cite{DBLP:journals/toct/DoronMVZ24}. The work of Spielman and Srivastava \cite{SS11} showed that, to construct a sparsifier, the following suffices: for each edge $e\in G$, include $e$ independently with probability $p_e$, where $p_e$ is proportional to the leverage score of $e$. Since (approximations to) $p_e$ are computable by $\polylog(n)$ depth circuits, their result immediately implies a randomized $\NC$ ($\RNC$) algorithm for sparsification. However, their result does not give a deterministic parallel algorithm, nor a catalytic logspace one.

Later, the work of Doron, Murtagh, Vadhan, and Zuckerman~\cite{DBLP:journals/toct/DoronMVZ24} gave a partial derandomization of this result. They showed that if the samples $p_e$ were generated from a $t$-wise independent distribution, one could obtain a sparsifier with $\tO(n^{1+1/t}\eps^{-2})$ edges. Since for any $t=O(1)$ we can sample from such a distribution in logspace, their work implies a $\searchCL$ algorithm for sparsifying a graph to $n^{1+\delta}$ edges for any constant $\delta>0$. However, they show this strategy \textit{cannot} give near-linear sparsity unless $t=\Omega(\log n)$, and hence they require $\Omega(\log^2 n)$ free space in this regime. Further work of Putterman, Vadhan, and Zaripov~\cite{PVZ26} analyzed bounded-independence graph sampling, but still does not imply a deterministic small-space algorithm.

\subsection{Proof Overview}
Our proof consists of two parts: analyzing leverage score sparsification using a pessimistic estimator, and going from such an estimator to a $\CL$ algorithm. We first discuss the former. 

\subsubsection{Building the pessimistic estimator via quasirandomness and the trace-exponential norm}\label{sec:trace-exp-description}

Before giving an overview of sparsification in $\CL$, we show how to construct a greedy sparsifier by tracking the change in a potential function over the graph as we fix edges of the sparsifier. We remark that this component of our proof is essentially implicit by combining the papers of Wigderson and Xiao~\cite{DBLP:journals/toc/WigdersonX08} and Spielman and Srivastava~\cite{SS11}. We give a direct proof (using similar ideas) that we feel better motivates the construction.

Our greedy sparsifier can be viewed as a \textit{derandomization} of Spielman and Srivastava's~\cite{SS11} method of constructing a spectral sparsifier via leverage-score sampling. Leverage-score sampling keeps each edge $e$ of the original graph independently with probability $q_e\approx p_e$, where $p_e$ is proportional to the leverage score of $e$, and reweights $e$ by $1/q_e$ when it is kept. This procedure requires randomness, and putting it in $\CL$ is subtle for two reasons. First, even with a deterministic greedy algorithm, it is not clear how to track the past decisions made in $\CL$, since we lack space to store a partial transcript.

Second, the quantity we would most naturally track (the operator norm of the normalized error) is neither a supermartingale under the sampling nor easily computable. Thus, we cannot derandomize against it directly. To fix this, we instead design a potential function $\Phi$ on \textit{partial} assignments $(\chi, S)$ where $S \subseteq E$ and $\chi \in \zo^S$, where $\chi_e \in \zo$ captures our decision to include/not include edge $e$ in the sparsifier. 
This potential function satisfies three critical properties:
\begin{enumerate}
    \item\label{itm:intro:good} For a complete state $\sigma=(\chi,E)$, if $\Phi(\sigma)$ is small (relative to $\Phi(\emptyset)$) then the graph constructed according to $\chi$ is a good sparsifier. 
    \item\label{itm:intro:comp} $\Phi(\sigma)$ is computable in $\CL$ (and $\P$).
    \item\label{itm:intro:super} For every $\sigma=(\chi,S)$ and $e\in E\setminus S$, let $\sigma_{e=b}$ be the state extended by including or not including edge $e$. Then
    \[
    \E_{b\sim \Bern(q_e)}[\Phi(\sigma_{e=b})]\le \Phi(\sigma)
    \]
    (and hence $\min_{b\in \zo}\Phi(\sigma_{e=b})\le \Phi(\sigma)$).
\end{enumerate}
The final property is exactly that $\Phi$ is a pessimistic estimator, as initially defined in \cite{DBLP:journals/jcss/Raghavan88}.

Given such a $\Phi$, we immediately obtain a greedy algorithm as follows: take any fixed order $e_1,\ldots,e_m$ over the edges of $G$, and at each step choose to include or not include $e$ to minimize the value of $\Phi$, using that it is efficiently computable by~\Cref{itm:intro:comp}. At the end of this process, we will have found a complete state $\sigma$ where $\Phi(\sigma)\le \Phi(\emptyset)$ by~\Cref{itm:intro:super}, which immediately implies that we have constructed a good sparsifier by~\Cref{itm:intro:good}. Thus, at least to obtain a greedy algorithm, it suffices to construct $\Phi$ with these properties.

To construct $\Phi$, let us reason about the quantity we must control. Recall (\Cref{eq:approxcond}) that $H$ is an $\ep$-spectral approximation of $G$ precisely when $\|L_G^{+/2}(L_H-L_G)L_G^{+/2}\|\le\ep$.\footnote{The $\ell_2$ operator norm of a matrix $M$, denoted $||M||$, is defined as: $||M|| := \sup_{x \neq 0} \frac{||Mx||_2}{||x||_2}$. If $M$ is a symmetric real matrix, $||M|| = \max_i |\lambda_i(M)|$.} Thus the quantity we must keep small is the $\ell_2$ operator norm of the error matrix. However, the operator norm itself is not a good potential: it is not a supermartingale under the random $\Bern(q_e)$ choices, and it is not obvious how to compute its expectation. 

\textbf{Bounding the operator norm via quasirandomness.} A more tractable way to certify that the operator norm is small is by upper-bounding it with even-powered trace moments and showing that these even-powered trace moments are small. 
The reason for this connects to the theory of quasirandom graphs~\cite{chung1989quasi, chung2002sparse}. \cite{chung2002sparse} shows that the count of length-$2t$ walks implies other deterministic properties possessed by random graphs. Thus, length-$2t$ closed walks (captured by the trace of the power) govern how ``random'' a graph looks in the quasirandomness literature~\cite{chung2002sparse} (i.e., how good a sparsifier of the complete graph a graph is). 
As in the complete graph setting \cite{chung2002sparse}, obtaining our desired sparsity requires $t=\log n$.

However, we again run into the issue of tractability in computing the even-power trace moments, so we instead use a smooth, convex proxy for the trace moments, which takes the form of a symmetrized trace exponential (see~\Cref{def:potential}).

\textbf{Constructing a pessimistic estimator via matrix inequalities.} Second, Lieb's inequality (\Cref{prop:lieb}) lets us pull the expectation over the next random choice \textit{inside} the trace, replacing the expectation of a function of a random matrix by a deterministic upper bound. This is what makes $\Phi$ a supermartingale. Such a potential is also called a \textit{pessimistic estimator} in the derandomization literature, because the terms capturing the not-yet-specified edges assume the worst over the remaining randomness, so greedy choices can only lower the potential.

Recall that constructing such a $\Phi$ does \textit{not} immediately give a $\searchCL$ algorithm, because in catalytic space we do not have space to store our previous set of choices. We show that we can \textit{combine} this pessimistic estimator with a second idea that we now explain.

\subsubsection{The compress--or--random technique}
We first overview the basic idea of compress--or--random. Suppose there is a \textit{randomized} algorithm   $A(r)=A(x,r)$ that, for every input $x$, returns a valid solution with high probability over the random string $r$ (and for simplicity, assume it either produces some fixed solution $y$ or a failure symbol $\perp$). Moreover, suppose $A$ is computable in $\CL$ given $x$ and $r$.\footnote{Despite the fact that randomized $\CL$ is known to be equal to deterministic $\CL$, it is not the case that this immediately gives a $\CL$ algorithm. This is because a randomized $\CL$ machine is given \textit{read-once} access to the string $r$, whereas $A$ is allowed to read each bit many times.} For now, think of our $A$ as a randomized algorithm that implements leverage-score sampling, though we will later re-introduce the potential function developed earlier.

The compress--or--random strategy treats the catalytic tape $\tau$ as the random string used by $A$. Once we do so, there are two cases:
\begin{enumerate}
    \item $A(\tau)$ produces a good output $y$. In this case, we can return $y$ and halt, and since $A$ is given read-only access to its randomness, we certainly restore $\tau$.
    \item $A(\tau)$ fails to produce a good output. Since $A$ succeeds whp, $\tau$ must be \textit{compressible} in an information-theoretic sense. 
    If we can perform this compression \textit{algorithmically}, then we will have freed up a large amount of space on the catalytic tape, which we can then use to solve the problem using a (space-inefficient) brute-force algorithm. Finally, we can decompress $\tau$ and successfully restore the tape.
\end{enumerate}
Unfortunately, for a generic randomized algorithm $A$ (even one computable in $\CL$) it is likely hard to give a compression algorithm for bad random strings (as such a result would imply derandomization consequences that are currently out of reach~\cite{LPT24}). Thus, in essentially every case, the most challenging part of materializing compress--or--random is giving an algorithmic compression scheme for the specific search problem of interest. We think of such a scheme as a pair $\Enc_A,\Dec_A$, where for every bad $\tau$ we have $|\Enc_A(\tau)|<|\tau|$, and $\Dec_A(\Enc_A(\tau))=\tau$.

\paragraph{The compression model}
Even worse, it does not suffice to find $\Enc,\Dec$ that are computable in $\CL$, in the sense that there is a $\CL$ machine that is given $\tau$ as input and prints $\Enc(\tau)$ to an output tape. Since we do not have separate workspace to write $\Enc(\tau)$ while computing it, our compression algorithm must \textit{transform} $\tau$ into $\Enc(\tau)$ (and the reverse for the decoding algorithm) with only $O(\log n)$ additional free space. Formally, we must show that both functions lie in the following class:
\begin{definition}[In-place catalytic logspace~\cite{DBLP:journals/corr/abs-2510-12005}]
    We say a length-preserving function $f:\zo^*\ra\zo^*$ is computable in $\inplaceFCL$ if the following holds. There is a machine with a read--write primary tape that holds $x$, an auxiliary worktape of length $O(\log |x|)$, and a catalytic tape holding $\tau\in \zo^{\poly(|x|)}$. For every $x$ and $\tau$, when the machine halts the primary tape holds $f(x)$, and the catalytic tape holds $\tau$. 
\end{definition}
Note that this is a \textit{stronger} condition than each bit of $\Enc$ being computable in $\CL$, since once we begin to modify $x$ into $y=\Enc(x)$, we lose access to the bits of $x$ that have been replaced with $y$.

The approach of prior compress--or--random algorithms~\cite{DBLP:journals/cc/Pyne25,DBLP:conf/focs/AgarwalaM25,DBLP:conf/innovations/AgarwalaAV26,CDKMS26} was to show for every bad $\tau$, there is a block of $\tau$ of length $C\log n$ that can be compressed given the rest of $\tau$ and $c\log n$ bits of auxiliary information, for $C>c$. Since we can store the $O(\log n)$ size block on the auxiliary worktape while searching for this information, this side-steps the challenge of in-place compression. Unfortunately, such a block-compression result does not appear to be true for leverage score sparsification: being a bad sparsifier is a global property that does not decompose in this fashion, and this is why the randomized algorithm of~\cite{SS11} does not directly give sparsification in $\CL$. 

Stepping back, we now show how our pessimistic estimator described in \Cref{sec:trace-exp-description} allows us to implement the compress--or--random technique. Specifically, we show how our pessimistic estimator can be used to build a compression scheme.

\subsubsection{Building compression from a potential function}
Our first observation is that the pessimistic estimator allows us to \say{score} prefixes of a candidate random string. For now, suppose all the sampling probabilities $q_e$ are equal to $1/2$. For a string $\chi\in \zo^k$ that one can think of as specifying the first $k$ include/not include decisions of the sparsifier, define
\[
f(\chi) := 2^{-k}\cdot\frac{\Phi(\chi)}{\Phi(\emptyset)}.
\]
Using the pessimistic estimator condition (\Cref{itm:intro:super}), it is easy to show that this $f$ is a \textit{semimeasure} on strings: 
\begin{definition}\label{def:semimeasure}
    We say $f:\zo^*\ra[0,1]$ is a \emphdef{semimeasure} if \[
    f(\emptyset)=1,
    \qquad
    f(u||0)+f(u||1)\le f(u)
    \quad\text{for all }u\in\{0,1\}^*.
    \]
\end{definition}
Moreover, the semimeasure given by $f$ is computable in $\CL$ since the potential $\Phi$ is. By Shannon's coding theorem~\cite{DBLP:journals/bstj/Shannon48}, we immediately obtain that there is an injective encoding function $\Enc$ where $|\Enc(\chi)|\le -\log f(\chi)+O(1)$. If we take a $\tau$ that gives a bad sparsifier (and hence $\Phi(\tau)\gg \Phi(\emptyset)$ by~\Cref{itm:intro:good}), it must be the case that 
\[|\Enc(\tau)|\le -\log f(\tau)+O(1)\le |\tau|-\Omega(\log n).\]
That is, such a $\tau$ is compressible by the encoding map. However, Shannon's coding theorem is again an existential statement. Our final step is to show that for a $\CL$-computable semimeasure $f$, there exists an in-place compression algorithm for $f$.

\subsubsection{In-place compression from a measure}
Recall that we are given a semimeasure $f$, and our goal is to construct in-place compression and decompression algorithms. Formally, we prove the following:
\begin{restatable}{theorem}{semicomp}\label{thm:semicomp}
    For every $C\ge 3$ the following holds. Let $f$ be a semimeasure computable in $\CL$.\footnote{We assume that the bitlength of $f(x)$ is bounded by $|x|^c$ for a constant $c$.} Then there is an encoding function $\oE:\zo^t\ra \zo^{t-C\log t}$ and $\inplaceFCL$ algorithms $\Enc,\Dec$ that work as follows.\footnote{If the semimeasure $f$ is computable given additional read-only information $G$, we assume $\Enc,\Dec$ have access to $G$ on an auxiliary read-only tape.}
    For every $x$ where $-\log f(x)<|x|-2C\log t$, we have that $\Enc(x)$ halts with the read-write tape equal to $\oE(x)||0^{C\log t}$. Moreover, $\Dec(\oE(x)||0^{C\log t})$ halts with the read-write tape equal to $x$.
\end{restatable}
We believe~\Cref{thm:semicomp} will have further applications in building catalytic algorithms. 

Our strategy for building $\Enc,\Dec$ is through arithmetic coding~\cite{DBLP:journals/cacm/WittenNC87}. Arithmetic coding associates an interval $I_x$ for every $x\in \zo^*$, with $|I_x|=f(x)$ (and the disjoint intervals $I_{x||0}$ and $I_{x||1}$ are nested inside $I_x$). In arithmetic coding, $x$ is then encoded by giving a string $y\in \zo^*$ that specifies a dyadic interval $D_y := [0.y,0.y+2^{-|y|})$ where $D_y\subseteq I_x$. It is straightforward to show that there always exists such a $y$ with $|y|\le -\log f(x)+2$. 

We show that there is an \textit{in-place transformation} between $x$ and such a $y$. For now, fix $x\in \zo^t$.
The most natural idea (and the one adopted by encoders used in practice~\cite{DBLP:journals/ibmrd/RissanenL79, DBLP:journals/cacm/WittenNC87, DBLP:journals/tois/MoffatNW98}) is to proceed for $i=1,\ldots,t$ and gradually replace $x_{1..i}$ with $y_{1..i}$. This has the benefit that it is always the case that $|y_{1..i}|\le |x_{1..i}|+O(1)$ (so after adding a constant amount of slack to the tape, we will not run out of buffer). However, we do not see a way to materialize this approach in-place in $\CL$. This is because after transforming the tape into
\[
y_{1..i} || x_{i+1..t},
\]
to determine the next bit $y_{i+1}$ of the encoding we must know the entire string $x$ (in particular, the prefix $x_{1..i}$ that we have already encoded), so it seems that we must \textit{decode} the prefix to determine the next bit. This results in linear computational depth and hence it is not clear how to implement it in $\CL$. 

We observe that we can materialize this strategy if we instead encode the \textit{suffix} of $x$. For $i=t,\ldots,1$, we transform $x$ into
\[
x_{1..i}||y_{i+1..t}.
\]
The most important observation is that the bits of $y$ corresponding to encoding $x_{1..i}$ can be determined from $x_{1..i}$ and a single extra bit (that essentially captures whether the interval will eventually go to the left or right of the midpoint of $I_{x_{1..i}}$). Thus, we do not need to decompress the suffix in order to determine the next bit of the encoding to output. 

We finally solve one additional issue. For a generic $f$ and $x$ with $-\log f(x)\ll |x|$, the string $y_{i+1..t}$ encoding $x_{i+1..t}$ can be arbitrarily longer than $x_{i+1..t}$, so we cannot replace it in-place without a large amount of scratch space. However, we observe that we \textit{can} perform this replacement if $x$ has the following property:
\begin{definition}[Informal]
    We say that $x$ has a \emphdef{final compression record (FCR)} if, letting $s(i) = i + \log f(x_{1..i})$, we have that $s(t)\ge s(i)$ for every $i\in [t]$.
\end{definition}
Note that $s(i)$ is roughly the number of bits saved by running arithmetic encoding on $x_{1..i}$. Thus, the property says that $x$ is most compressible if we compress it in its entirety. We show that for these strings, we can perform this in-place replacement using constant slack:
\begin{proposition}[Informal]
    Suppose $x$ has a final compression record. Then the bits $y_{i..t}$ encoding $x_{i..t}$ satisfy $|y_{i..t}|\le |x_{i..t}|+O(1)$ for every $i$.
\end{proposition}
\say{The bits encoding $x_{i..t}$} is nontrivial to formally define, but this captures the main idea. Finally, we still must compress $x$ that do not have an FCR. To do so, we take the first index $i$ where $x_{1..i}$ is compressible by at least $C\log t$ bits (and hence $x_{1..i}$ must have an FCR) and invoke our in-place encoder only on this prefix (and store the remainder of the string unmodified). This suffices to free up the claimed number of bits, so we are done. 

\subsubsection{Putting it all together}
Finally, we informally overview the ultimate construction, which follows the conventional recipe of a compress--or--random algorithm. We think of the catalytic tape $\tau$ as being divided into $\tau_1,\ldots,\tau_{m}$, where $\tau_i\in \zo^E$ specifies a candidate sparsifier. If for some $i$, $\Phi(\tau_i)/\Phi(\emptyset)$ is small, we output the corresponding sparsifier without modifying the catalytic tape. Otherwise, we use~\Cref{thm:semicomp} with the measure $f$ defined in terms of $\Phi$ to compress $\tau_i$ into $\oE(\tau_i)||0^{C\log n}$ for every $i$. After doing so and rearranging the tape, we obtain $\Omega(n\log n)$ free bits on the catalytic tape. Using this large amount of free space, we find a sparsifier by brute force, output it, then finally decode $\oE(\tau_i)$ to $\tau_i$ for every $i$, restoring the tape.

\subsection{Open questions}
It is a well-known open problem to determine the relative power of catalytic space and parallel algorithms (i.e. $\NC$ versus $\CL$)~\cite{DBLP:journals/eatcs/Mertz23}. In the specific case of compress--or--random, techniques applicable to one often imply or draw on results for the other~\cite{DBLP:journals/cc/Pyne25,DBLP:conf/focs/AgarwalaM25}. Thus, in our view, an especially interesting open question is to understand the randomness complexity of \textit{parallel} algorithms for sparsification. Recall that~\cite{SS11} gives a randomized parallel algorithm for this problem. Every prior compress--or--random result places the corresponding search problem in the class $\LOSSY[\NC]\subseteq \ZPNC$~\cite{DBLP:conf/focs/Korten21}, because the compression and decompression functions $\Enc_A,\Dec_A$ are computable in $\NC$. Interestingly, this does \textit{not} appear to be true here: we do not see a way to implement the decoder in $\NC$. This is because the measure is essentially given as a conditional oracle (where given $x_{1..i}$ we can determine the size of intervals $I_{x_{1..i||0}}$ and $I_{x_{1..i||1}}$), and hence we must decode the first $i$ symbols to learn how to decode the $i+1$st.

\subsection{Roadmap}
In~\Cref{sec:prelims} we define relevant complexity classes and give useful facts about matrices and graphs. In~\Cref{sec:leverage} we build the potential function $\Phi$. In~\Cref{sec:measure} we show how to go from $\Phi$ to a measure. In~\Cref{sec:semicomp} we show how to compress in-place using a $\CL$-computable measure. In~\Cref{sec:final} we combine these ingredients and prove~\Cref{thm:main}.

\section{Preliminaries}\label{sec:prelims}
For binary strings $x_{<},x\in \zo^*$ where $x_<$ is a prefix of $x$, we let $x-x_{<}$ denote the (possibly empty) suffix of $x$ not contained in $x_{<}$.

For a symmetric matrix $M$ we
write $\|M\| = \sup_{\|x\|_2 = 1} |x^{\intercal} M x|$ for the $\ell_2$-operator norm, and $\preccurlyeq$ for the Loewner
order.

\subsection{Catalytic Logspace}
We define the notion of functions computable in small space.
\begin{definition}[Functions in catalytic logspace]
    For a function $f:\zo^*\ra\zo^*$, we say that $f$ is computable in catalytic logspace if the map $(x,i)\ra f(x)_i$ is computable in $\CL$. 
\end{definition}

\subsection{Graphs and Laplacians}

For a (weighted or unweighted) graph
$G$ on vertex set $[n]$ let
\[L_G := D_G - A_G\] 
denote its Laplacian. 
Let $L_G^+$ denote the Moore--Penrose pseudoinverse of $L_G$, and let $L_G^{+/2}$ be the square root of the pseudoinverse. Note that for graphs $G,H$, we have that $H$ is an $\ep$-spectral approximation of $G$ if $\ker L_G\subseteq \ker L_H$ and
\begin{equation}\label{eq:approxcond}
\|L_G^{+/2}(L_H-L_G)L_G^{+/2}\|\leq \ep.
\end{equation}
Let
$\Pi_G := L_G^{+/2} L_G L_G^{+/2}$ be the orthogonal projection onto
$\mathrm{im}(L_G)$; when $G$ is clear from context, we refer to this as $\Pi$. Note that if the graph is connected, $\Pi = I - \tfrac1n \one\one^{\intercal}$.

\paragraph{Rank-$1$ Graph Decomposition}
For a pair
$e = \{x, y\} \in \binom{[n]}{2}$, let $b_e := e_x - e_y \in \R^n$, so that for a graph $G$ we have
$L_G = \sum_{e \in E(G)} b_e b_e^{\intercal}$.

\paragraph{Exponential inequality}
We state a useful inequality for the exponential.
\begin{fact}[Polynomial-exponential comparison]\label{fact:poly}
For every $u \in \R$, every integer $t \geq 1$, and every $c > 0$,
$$u^{2t} \leq \frac{(2t)!}{2 c^{2t}} \left( e^{c u} + e^{-c u} \right).$$
\end{fact}

\begin{proof}
$e^{cu} + e^{-cu} = 2 \sum_{k \text{ even}} \frac{(cu)^k}{k!}
\geq 2 \frac{(cu)^{2t}}{(2t)!}$, since all terms are nonnegative.
\end{proof}

\subsection{Matrix Preliminaries}

We say a matrix $M\in \F^{n\times n}$ is positive semidefinite (PSD) if $x^\top Mx\ge 0$ for every $x\in \F^n$. We use $A \preccurlyeq B$ to denote that $B-A$ is PSD. 

\begin{definition}[Matrix exponential]
For a matrix $M\in \F^{n\times n}$, we define the matrix exponential $e^M\in \F^{n\times n}$ as
$$e^M := \sum_{i\ge 0}\frac{M^i}{i!}.$$
\end{definition}

\begin{definition}[Matrix logarithm]
    If $M$ is real symmetric and positive definite, let 
$$M = U \mathrm{diag}(\lambda_1, \lambda_2, \dots, \lambda_n) U^{\top}, ~~ \lambda_i > 0,$$
be its spectral decomposition. Define its principal matrix logarithm by:
$$\log M := U \mathrm{diag}(\ln \lambda_1, \ln \lambda_2, \dots, \ln \lambda_n) U^{\top}.$$
\end{definition}

Throughout, scalar $\log$ denotes $\log_2$ while $\ln$ denotes natural logarithm. The matrix logarithm uses natural logarithms of the eigenvalues.

We now list some useful facts.

\begin{fact}\label{fct:normpower}
    For a symmetric matrix $M$, $\|M\|^{2t} =\max_i\{\lambda_i(M)^{2t}\}$.
\end{fact}
\begin{fact}[Convexity of the matrix trace exponential (from Klein's lemma)]\label{fct:TrExpconvex}
    The function $M\ra \Tr(\exp(M))$ is convex for symmetric matrices $M$.
\end{fact}
\begin{fact}[Perturbation of trace powers]\label{fct:trace-perturb}
    Let $A,B$ be $n\times n$ real matrices with $\|A\|,\|B\|\le T$ for some
    $T\ge 1$. Then for every $k\in\N$,
    \[
    \left|\Tr(A^k)-\Tr(B^k)\right| \le k n T^{k-1}\cdot \|A-B\|.
    \]
\end{fact}
\begin{proof}
    Telescoping, $A^k - B^k = \sum_{j=0}^{k-1} A^{j}(A-B)B^{k-1-j}$, and each
    summand has norm at most $T^{k-1}\|A-B\|$ by submultiplicativity, and then since $|\Tr(C)| \le n\|C\|$ the claim follows.
\end{proof}

We note the following fact regarding trace exponentials. 
\begin{fact}[Monotonicity of the trace exponential]\label{fact:mono}
If $A \preccurlyeq B$ then $\Tr \exp(A) \leq \Tr \exp(B)$.
\end{fact}

\begin{proof}
By the Courant--Fischer theorem, $A \preccurlyeq B$ implies
$\lambda_k(A) \leq \lambda_k(B)$ for every $k$, and
$\Tr \exp(\cdot) = \sum_k e^{\lambda_k(\cdot)}$ with $e^u$ increasing.
\end{proof}

We now give Lieb's inequality, which is a key technical tool we rely on.
\newcommand{\bH}{\mathbf{H}}
\newcommand{\bY}{\mathbf{Y}}

\begin{proposition}[Corollary of Theorem 6 of \cite{lieb1973convex}, as in \cite{Tropp12}]\label{prop:lieb}
Let $H$ be a fixed symmetric matrix and let $X$ be
a random symmetric matrix. Suppose $Y$ is a fixed
symmetric matrix satisfying $\ln \E e^{X} \preccurlyeq Y$. Then
$$\E \left[ \Tr \exp\left( H + X \right) \right]
  \leq \Tr \exp\left( H + Y \right).$$
\end{proposition}

To simplify the use of Lieb's inequality, which involves any symmetric matrix $Y$ satisfying $\ln \E e^{X} \preccurlyeq Y$, we will rely on the following bounded Bernstein mgf inequality.

\begin{proposition}[Bounded Bernstein mgf, as in \cite{Tropp12}, after scaling]
Suppose that $X$ is a random symmetric matrix with $\E[X] = 0$ and
$\lambda_{\max}(X) \leq a$ almost surely. Then, for every $\theta >0$,
$$\ln \E\left[ e^{\theta X} \right]
  \preccurlyeq \frac{h(a \theta)}{a^2} \cdot \E\big[X^2\big],
  \qquad \text{where } h(u) := e^u - u - 1.$$
\end{proposition}
Using the fact that $h(u)\le u^2$ on $[0,1]$ we obtain a user-friendly corollary.
\begin{corollary}[Simplified Bernstein bound]\label{cor:bernstein}
Suppose that $X$ is a random symmetric matrix with $\E[X] = 0$ and
$\lambda_{\max}(X) \leq a$ almost surely. Then, for every $\theta \in (0,1/a)$,
$$\ln \E\left[ e^{\theta X} \right]
  \preccurlyeq \theta^2\cdot \E\big[X^2\big].$$
\end{corollary}

\section{Derandomized leverage-score sparsification}\label{sec:leverage}

For this section, we always denote the graph to be sparsified by $G$. We assume without essential loss of generality that  $G$ is connected, as otherwise we can apply the algorithm independently on each component. A \emph{state} is a pair $\sigma = (S, \chi)$ where $S$ is a set of decided
edges and $\chi \in \{0,1\}^{S}$, where $\chi_e=1$ if we include the edge $e$ in the sparsifier $H$. We say a state is complete if $S=E$.

The standard randomized leverage score sampling keeps edge $e$ with probability $p_e$ and,
when kept, gives it weight $1/p_e$, for probabilities $p_e$ that depend on each edge.
We show a potential function that enables us to measure \say{how well leverage score sampling is going} as we progress over the list of edges.

To do this, we define 
\[
\Lambda_e := L_G^{+/2} b_e b_e^{\intercal} L_G^{+/2},\qquad r_e := \Tr(\Lambda_e)=R_{\mathrm{eff}}(e)
\]
where the second quantity is, by definition, the \emphdef{effective resistance} of $e$. Note that $\Lambda_e$ is always a PSD rank-one matrix, and we have the following two useful facts
\begin{equation}\label{eq:ell-facts}
  \Lambda_e^2 = r_e \Lambda_e
  \qquad \text{and} \qquad
  \sum_{e \in E} \Lambda_e = L_G^{+/2}\Big( \sum_e b_e b_e^{\intercal} \Big) L_G^{+/2}
  = L_G^{+/2} L_G L_G^{+/2} = \Pi 
\end{equation}
and hence
\begin{equation}\label{eq:foster}
  \sum_{e \in E} r_e = n-1.
\end{equation}
Next, fix a parameter $s$ (to be decided later) and set 
\begin{equation}\label{eq:pe}
  p_e := \min\big(1, s r_e\big).
\end{equation}
In fact, our algorithm works with small multiplicative over-estimates of the probabilities $p_e$, which enables working with approximate leverage scores without substantially affecting the number of edges or approximation quality. 
\begin{definition}
    We say a set of sampling probabilities $\{q_e\}_{e\in E}$ are \emphdef{good} if $q_e\in [p_e,3p_e]$, $q_e\le 1$, and $q_e = 2^{-j_e}$ for some $j_e\in \N$ for every $e$.
\end{definition}
Note that for any set of true probabilities $\{p_e\}_{e\in E}$, a corresponding set of good probabilities exists since we can round up $p_e$ to the next power of two. We require good probabilities to be powers of two purely to aid the compression argument, and the potential function is well defined without this condition.
We say $e$ is a \textit{randomized} edge if $q_e<1$. 

Every non-randomized edge will always be included with weight $1$.
We thus define the (weighted) sparsifier as follows, for a complete state $\sigma$
\begin{equation}\label{eq:Hdef}
    L_{H(\sigma)} := \sum_{e \in E} \frac{\chi_e}{q_e} b_e b_e^{\intercal}.
\end{equation}

\begin{theorem}[Pessimistic estimator for graph approximation]\label{thm:pot}
    Given a graph $G=(V,E)$ on $n$ vertices and a set of good sampling probabilities $\{q_e\}_{e\in E}$ and $\ep\in (0,1/2)$, there is a potential function $\Phi$ defined over states with the following properties.
    \begin{enumerate}
        \item\label{itm:comp} $\Phi(\sigma)$ is computable in $\CL$ for any (partial or complete) state $\sigma$.
        \item\label{itm:impliesgood} For every integer $K \geq 1$ where $2^K\cdot \ep< 1/2$ and complete state $\sigma$ where $\Phi(\sigma) \le n^{K}\Phi(\emptyset)$, we have that the graph $H(\sigma)$ defined by $\sigma$ is an $\ep\cdot 2^K$-spectral sparsifier of $G$, and contains $O(n\ep^{-2}\log n)$ edges. 
        \item\label{itm:super} For every $e\in E$ and state $\sigma=(S,\chi)$ where $e\notin S$, let
        $\sigma \cup (e \mapsto b)$ be the extended state where $\chi_e=b$. Then
        \[\E_{b \sim \Bern(q_e)} \big[ \Phi(\sigma \cup (e \mapsto b)) \big]
        \leq \Phi(\sigma).\]
    \end{enumerate}
\end{theorem}

\subsection{Analyzing the potential function}
We first define the function $\Phi$. We assume WLOG that $\ep = 2^{-\ell}$ for $\ell\in \N$, and assume WLOG that $\log n=\log_2 n$ is an integer. We set 
\[
s := \frac{16}{\ep^2} \log n , \qquad a:=\frac{\log n}{s},\qquad t := \log n,
  \qquad c := 4\ep^{-1}, \qquad D := \frac{(2t)!}{2 (c\log n)^{2t}}.
\]
Recall $q\in [0,1]^E$ is the set of good sampling probabilities. Given a state $(S,\chi)$, we first define the matrices
\[
X_e := \log n \cdot \left(\frac{\chi_e}{q_e}-1\right)\Lambda_e,\qquad Y_e := \log n \cdot \Lambda_e.
\]
Observe that $X_e$ is only well-defined when $e \in S$, but $Y_e$ is well-defined for all $e \in E$.

Next, define:
\[M_\sigma := \sum_{e \in S} X_e,
  \qquad V_\sigma := \sum_{e \in E\setminus S} Y_e.\]

Intuitively, $M_\sigma$ represents the centered contribution of the
decided slots (note that a slot decided to $0$ contributes the negative matrix
$-\Lambda_e$, not zero), and $V_\sigma$ represents an upper bound on the potential of the undecided slots. We can then define the potential. 
  \begin{definition}[Potential function]\label{def:potential}
$$\Phi(\sigma) := 
  \Tr \exp\big( c M_\sigma + V_\sigma \big)
      + \Tr \exp\big( -c M_\sigma + V_\sigma \big).$$
\end{definition}
We now prove each item of~\Cref{thm:pot}.

\subsubsection{The potential is a supermartingale}
In this section we prove~\Cref{itm:super}. We first establish some useful facts about the matrix $X_e$, viewed as a random matrix over $\chi_e\sim\Bern(q_e)$.

\begin{lemma}\label{lem:slot-lev}
Let $e\in E$ be arbitrary. We have 
\begin{enumerate}
  \item $\E_{\chi_e\sim \Bern(q_e)}[X_e] = 0$;
  \item $\lambda_{\max}(X_e) \le a$ and $\lambda_{\max}(-X_e) \le a$, almost surely;
  \item $\E_{\chi_e\sim\Bern(q_e)}\big[X_e^2\big] \preccurlyeq \frac{(\log n)^2}{s} \Lambda_e = c^{-2}\cdot Y_e$. Hence summing over all edges,
  $$\sum_{e \in E} \E_{\chi_e\sim\Bern(q_e)}\big[X_e^2\big]
    \preccurlyeq \sum_{e \in E} \frac{(\log n)^2}{s} \Lambda_e
    \preccurlyeq \frac{(\log n)^2}{s} \Pi .$$
\end{enumerate}
\end{lemma}
\begin{proof}~
All facts are immediate if $e$ is not a random edge (i.e. $q_e=1$), so we WLOG assume it is.
\begin{enumerate}
    \item This is immediate from $\E[\chi_e/q_e] = 1$.
    \item The nonzero eigenvalue of
    $X_e$ is $\log n \cdot \big(\tfrac{\chi_e}{q_e} - 1\big) r_e$. When $\chi_e = 1$ this equals
    \[
    \log n \cdot \frac{1 - q_e}{q_e} r_e \le \log n \cdot \frac{1}{p_e}\cdot r_e \le \frac{\log n}{s} = a
    \]
    using
    $q_e\ge p_e $ and $p_e= s r_e$; when $\chi_e = 0$ it equals $-\log n \cdot r_e$. Thus
    $\lambda_{\max}(X_e)\le a$, and
    $\lambda_{\max}(-X_e) = \log n\cdot r_e\le \log n/s=a$ since $p_e = s r_e < 1$.
     \item First note that $\mathrm{Var}(\chi_e/q_e) = \mathrm{Var}(\chi_e)/q_e^2=(1-q_e)/q_e$, and $\Lambda_e^2 = r_e \Lambda_e$ from \Cref{eq:ell-facts}. Thus
    $$\E[X_e^2] =(\log n)^2\cdot \frac{1 - q_e}{q_e} r_e \Lambda_e \preccurlyeq \frac{(\log n)^2}{p_e} r_e \Lambda_e =\frac{(\log n)^2}{s} \Lambda_e = c^{-2}\cdot Y_e.$$ Then taking a sum over the randomized edges gives the final claim, where the final step uses \Cref{eq:ell-facts} (and that $\Lambda_e$ is PSD for every $e$). \qedhere
\end{enumerate}
\end{proof}

We next observe that the matrices $Y_e$ serve as an upper bound for $X_e$ in the potential function.
\begin{claim}\label{clm:Ye}
    We have $\log \E_{\chi_e\sim \Bern(q_e)} e^{c X_e} \preccurlyeq Y_e$ and $\log \E_{\chi_e\sim \Bern(q_e)} e^{-c X_e} \preccurlyeq Y_e$.
\end{claim}
\begin{proof}
    We have that
    \begin{align*}
        \log \E_{\chi_e\sim \Bern(q_e)} e^{c X_e} &\preccurlyeq c^2 \E[X_e^2] && \text{(\Cref{cor:bernstein})}\\
        &\preccurlyeq Y_e && \text{Item 3 of~\Cref{lem:slot-lev}}
    \end{align*}
    where the first step uses that $\E[X]=0$ and $c<1/a$ and $\lambda_{\max}(X_e)\le a$ from Item 1 and 2 of~\Cref{lem:slot-lev}. The second bound is identical, using that $\E[(-X_e)^2] = \E[X_e^2]$.
\end{proof}

We can then prove~\Cref{itm:super}.
\begin{proof}[Proof of~\Cref{itm:super}]
First note that fixing $e \in E\setminus S$ to the value $\chi_e$ effectively replaces $V_\sigma$ by
$V_\sigma - Y_e$ and adds the matrix $X_e$ to $M_\sigma$. Set
$\mathbf{H}_{\pm} := \pm c M_\sigma + V_\sigma - Y_e$, which are fixed matrices that do not depend on $\chi_e$. We first consider the positive exponent in the potential
\[
    \E_{\chi_e\sim \Bern(q_e)} \Tr \exp\big( \mathbf{H}_{+} + c X_e \big) \le \Tr \exp\big( \mathbf{H}_{+} + Y_e \big) =\Tr \exp\big( c M_\sigma + V_\sigma \big) 
\]
where the inequality follows from~\Cref{prop:lieb} and~\Cref{clm:Ye}. The equivalent bound for the negative exponent is identical, and adding the two branches yields the claim.
\end{proof}

\subsubsection{Small potential implies sparsification}
In this section we prove that for a complete state $\sigma$ with a potential that is small (within a $\poly(n)$ factor) of the initial potential, the resulting graph is a good sparsifier (both in the sense that it spectrally approximates $L_G$ and has few edges).

The first component of the proof is showing the initial empty state has small potential.
\begin{lemma}\label{lem:pot-initial}
The empty state satisfies $D\cdot \Phi(\emptyset) \le \ep^{2t}$. 
\end{lemma}
\begin{proof}
First note that
\begin{equation}\label{eq:V-bound-gen}
  V_\emptyset  = \sum_{e\in E}Y_e = \log n \sum_{e\in E}  \Lambda_e = \log n \cdot \Pi \preccurlyeq I\cdot \log n
\end{equation}
and so we have
\[\Phi(\emptyset) = 2\Tr \exp\big( V_\emptyset \big) \le 2\Tr \exp\big( \log n I \big) \le 2n^{1+1/\ln 2}.\]
where the first inequality follows from~\Cref{eq:V-bound-gen} and~\Cref{fact:mono}. 
It remains to bound
\[
\frac{D\Phi(\emptyset)}{\ep^{2t}}\le n^{1+1/\ln 2}\cdot \frac{(2t)!}{\left(4t\right)^{2t}}
\]
by $1$. 
Using the standard bound $k! \le e\sqrt{k}(k/e)^{k}$
with $k = 2t$,
\[
\frac{(2t)!}{(4t)^{2t}}
\le e\sqrt{2t}\left(\frac{2t}{4et}\right)^{2t}
= e\sqrt{2t}(2e)^{-2t}.
\]
Finally $(2e)^{-2t} = 2^{-2t(1+\log e)} = n^{-2(1+1/\ln 2)}$, and hence we have that the bound holds.
\end{proof}

We next show that the potential of a complete state upper bounds the norm of the term $M_\sigma/\log n$, which captures the quality of the sparsifier by~\Cref{eq:approxcond}. This is exactly where the \textit{analogy to graph quasirandomness} comes in. We observe that
$$\left\|\frac{1}{\log n}M_\sigma\right\|^{2t} \leq \Tr((M/\log n)^{2t})$$
which is a \textit{quasirandom-style condition} because $\Tr((M/\log n)^{2t})$ is an even-power trace moment that can be viewed as a weighted and signed closed-walk statistic of the corresponding graph. This is reminiscent of a quasirandomness notion of \cite{chung1989quasi, chung2002sparse}, where being a good sparsifier of the complete graph (i.e., looking ``random enough'') is implied by the count of length-$2t$ walks in the graph.

Given this connection, we are ready to state and prove the bound on the norm.

\begin{lemma}\label{lem:pot-final-bound}
Let $\sigma$ be a complete state and let $H$ be the sparsifier as defined in~\Cref{eq:Hdef}. 
Then
$$\left\|\frac{1}{\log n}M_\sigma\right\|^{2t} \le D\cdot \Phi(\sigma) .$$
\end{lemma}
\begin{proof}
Let $M:=M_\sigma$ for convenience.
Note that it suffices to prove 
$ \Tr((M/\log n)^{2t})\le D\cdot \Phi(\sigma)$, because $\left\|\frac{1}{\log n}M_\sigma\right\|^{2t} =\max_i\{\lambda_i(M_\sigma)/\log n\}^{2t}\le \sum_i (\lambda_i(M_\sigma)/\log n)^{2t}=\Tr((M/\log n)^{2t})$ where the first equality is~\Cref{fct:normpower}. Apply \Cref{fact:poly} to each eigenvalue
$\lambda_i/\log n$ of $M/\log n$ with parameters $t$ and $c\log n$ and sum:
\begin{align*}
    \Tr\big((M/\log n)^{2t}\big) &= \sum_i (\lambda_i/\log n)^{2t}\\
  &\leq \frac{(2t)!}{2 (c\log n)^{2t}} \sum_i \big( e^{c \lambda_i} + e^{-c \lambda_i} \big)\\
  &= \frac{(2t)!}{2 (c\log n)^{2t}} \big( \Tr e^{c M} + \Tr e^{-c M} \big) = D\cdot \Phi(\sigma),
\end{align*}
using that $\exp(\pm c M)$ has eigenvalues $e^{\pm c \lambda_i}$, and the final equality holds since $\sigma$ is complete and hence $V_\sigma=0$. The first
inequality holds since all summands $(\lambda_i/\log n)^{2t}$ are nonnegative.
\end{proof}

We likewise prove that any complete state that is a good sparsifier is also sparse:
\begin{lemma}[Edge count]\label{lem:edge-count}
Let $\sigma$ be a complete state where $L_H \preccurlyeq (1+\gamma) L_G$. Then
$$|E(H)| \le (1+\gamma) 3sn.$$
\end{lemma}
\begin{proof}
Write $w_e := \chi_e/q_e$ for the weight $H$ assigns to $e$, so that 
$L_H = \sum_e w_e b_e b_e^{\intercal}$, and recall
$r_e = \Tr(L_G^+ b_e b_e^{\intercal})$. We first claim the per-edge bound
\begin{equation}\label{eq:edge-ineq}
  \chi_e \le 3 s r_e w_e \qquad \text{for every } e.
\end{equation}
This is equivalent to $q_e\le 3sr_e$, which we know is true since $q_e\le 3p_e \leq 3sr_e$ because the algorithm outputs a good set of probabilities $\{q_e\}_{e \in E}$.
Summing
\Cref{eq:edge-ineq} and using $r_e w_e = \Tr(L_G^+ w_e b_e b_e^{\intercal})$ we obtain that $\sum_{e\in E}r_ew_e= \Tr\!\big(L_G^+ L_H\big)$, and hence
$$|E(H)| = \sum_e \chi_e \le 3s \sum_e r_e w_e
  = 3s\Tr\!\big(L_G^+ L_H\big).$$
Since $L_H \preccurlyeq (1+\gamma)L_G$ and $L_G^+ \succcurlyeq 0$, the trace inequality $\Tr(CA) \le \Tr(CB)$ for $0 \preccurlyeq A \preccurlyeq B$ and
$C \succcurlyeq 0$ gives
\[|E(H)|\le 3s\Tr\!\big(L_G^+ L_H\big)
  \le 3s(1+\gamma) \Tr\!\big(L_G^+ L_G\big)
  = (1+\gamma) 3s \Tr \Pi
  \le (1+\gamma) 3sn.\qedhere\]
\end{proof}

Finally, we can combine these elements and prove~\Cref{itm:impliesgood}.
\begin{proof}[Proof of~\Cref{itm:impliesgood}]
    Let $\sigma$ be a complete state where $\Phi(\sigma) \le \Phi(\emptyset)\cdot n^K$ and thus $D\Phi(\sigma) \le D\Phi(\emptyset)\cdot n^K$. By~\Cref{lem:pot-initial}, we have $D\Phi(\sigma)\le \ep^{2t}\cdot n^K$, and hence by~\Cref{lem:pot-final-bound} we have $\left\|\frac{1}{\log n}M_\sigma\right\|^{2t} \le \ep^{2t}\cdot n^K$, i.e.
    \begin{align*}
        \left\|L_G^{+/2}(L_H-L_G) L_G^{+/2}\right\| &= \left\| \frac{1}{\log n}M_{\sigma}\right\|\\
        &\le \ep\cdot n^{K/2t}=\ep\cdot 2^{K/2}
    \end{align*}
    and hence $H(\sigma)$ is an $\ep\cdot 2^{K/2}$-spectral approximation of $G$ by~\Cref{eq:approxcond}. Finally, by~\Cref{lem:edge-count} applied with $\gamma=\ep \cdot 2^{K/2}$, we have that 
    \[
    |E(H)| \le 48\left(1+\ep\cdot 2^{K/2}\right)\frac{n\log n}{\ep^2} = O(n \log n/\ep^2)
    \]
    as claimed.
\end{proof}

\subsection{Computing the potential in catalytic logspace}
Finally, we show~\Cref{itm:comp}. We actually show that we can approximate the potential to very small error, then define a modified potential $\Phi'$ based on this approximation and still show it satisfies~\Cref{itm:impliesgood} and~\Cref{itm:super}.

We require the fundamental result of~\cite{DBLP:conf/stoc/BuhrmanCKLS14} that matrix powering is computable in $\CL$. 
\begin{theorem}[\cite{DBLP:conf/stoc/BuhrmanCKLS14}]\label{thm:BCKLS}
    Evaluation of logspace-uniform $\TC^1$ circuits is in $\CL$. In particular, given an $n\times n$ matrix $M$ with $\poly(n)$ bit entries and $k$ in unary, computing $M^k$ is in $\CL$.
\end{theorem}

\begin{claim}\label{clm:compLperp}
    There is a $\CL$ algorithm that, given $G$ and $N$ in unary, outputs a symmetric matrix $\widehat{L}^{+}$ with $\poly(n)$-bit rational entries and $\big\|\widehat{L}^{+}-L_G^{+}\big\|\le 2^{-N}$.
\end{claim}
\begin{proof}
    WLOG assume $G$ is connected (as otherwise we can apply the algorithm to each connected component). Write $L:=L_G$ and $M:=I-\tfrac{L}{2n}$. As $G$ is simple, $\lambda_{\max}(L)\le 2(n-1)<2n$, so
    $M$ is symmetric with $M\one=\one$ and, on $\one^{\perp}$, eigenvalues in $[\tfrac1n,1-\gamma]$
    where $\gamma:=\lambda_{\min}^{+}(L)/(2n)$. A standard spectral-gap bound gives
    $\lambda_{\min}^{+}(L)\ge 1/n^{2}$ for connected $G$, and hence $\gamma\ge \frac{1}{2n^{3}}$.
    
    Since $\Pi$ commutes with $M$ and $L=2n(I-M)$ is invertible on $\one^{\perp}$,
    \[
    L^{+}=\frac{1}{2n}\Pi\Big(\sum_{k\ge 0}M^{k}\Big)\Pi ,
    \]
    the series is converging on $\one^{\perp}$ as $\|M|_{\one^{\perp}}\|\le 1-\gamma<1$. A standard analysis shows that we can truncate the power series at 
    $K:=\poly(n)\cdot N$ terms and obtain error $2^{-N}$, so we output a $\poly(n)$-bit rounding of $\widehat{L}^{+}:=\tfrac{1}{2n}\Pi\big(\sum_{k=0}^{K}M^{k}\big)\Pi$. Since $M$ can be explicitly computed given $G$ in logspace and matrix powering with $\poly(n)$-bit entries is in $\TC^1$ and hence $\CL$ by~\Cref{thm:BCKLS}, we obtain the desired result by composition.
\end{proof}
We recall a basic fact that the sampling probabilities are nontrivial. Recall that $p_e := \min\big(1, s r_e\big).$

\begin{fact}\label{fct:pe-nonneg}
    For a simple graph $G$, we have $p_e \ge 1/n$ for every $e\in E$.
\end{fact}

Then we have that we can compute a good set of probabilities by approximating $L_G^+$. Recall that a good set of probabilities $\{q_e\}_{e\in E}$ satisfies $q_e \in [p_e, 3p_e], q_e \leq 1,$ and $q_e = 2^{-j_e}$ for some $j_e$ for each $e$.
\begin{corollary}\label{cor:CLqe}
    There is a $\CL$ algorithm that, given a graph $G$, outputs a good set of probabilities $\{q_e\}_{e\in E}$.
\end{corollary}
\begin{proof}
    Since $r_e = b_e^{T} L^+ b_e$, set estimate $\widehat{r}_e = b_e^{T} \widehat{L}^+ b_e$. Because $||b_e||^2 = 2$, we have that $|\widehat{r}_e - r_e|  \leq 1/(4n)$ (using $N = \lceil \log_2(8n) \rceil$ in \Cref{clm:compLperp}):
    $$|\widehat{r}_e - r_e| = |b_e^{T} (\widehat{L}^+ - L^+) b_e| \leq ||b_e||^2 ||\widehat{L}^+ - L^+|| \leq \frac{1}{4n}.$$
    
    Adding $1/(4n)$ to the estimate, we obtain an estimate $\widetilde{r}_e$  of $r_e$ satisfying $r_e \leq \widetilde{r}_e \leq \frac{3}{2} r_e$, where the last inequality uses $r_e \geq 1/n$.

    Thus, using $\widetilde{p}_e = \min (1, s \widetilde{r}_e)$ as the estimate, we have $p_e \leq \widetilde{p}_e \leq \frac{3}{2} p_e$. Let $q_e$ be the smallest number of the form $2^{-j}, j \in \mathbb{N}$ satisfying $q_e \geq \widetilde{p}_e$. Then 
    $$p_e \leq q_e < 2 \widetilde{p}_e \leq 3 p_e$$
    and $q_e \leq 1$. Thus, $\{q_e\}_{e \in E}$ is good. The construction of $\{q_e\}_{e \in E}$ is in $\CL$ because of \Cref{clm:compLperp}, and since all remaining operations can be performed in logspace. 
\end{proof}

With this we can prove that the potential function can be computed in $\CL$.

\begin{claim}\label{clm:Wnorm}
Let $\sigma=(S,\chi)$ be any state and $r\in\zo$, and define
\[
W_r := (-1)^r\cdot c\log n\sum_{e\in S}\Big(\tfrac{\chi_e}{q_e}-1\Big)b_eb_e^\intercal
+ \log n\sum_{e\notin S}b_eb_e^\intercal .
\]
Then $W_r$ is an integer matrix with entries of magnitude at most
    $10n^5$ (and $\|W_r\|\le 10n^5$), and there is a logspace algorithm that, given $G$, $\{q_e\}$,
    $\sigma$, and $r$, outputs $W_r$.
\end{claim}
\begin{proof}
For the first item, recall that $c\log n$ and $\log n$ are integers, and each
$1/q_e = 2^{j_e}$ are integers, and each $b_eb_e^\intercal$ is an integer
matrix, so $W_r$ is an integer matrix. For the magnitude, each coefficient
satisfies $|\chi_e/q_e - 1| \le 1/q_e \le 1/p_e \le n$
(\Cref{fct:pe-nonneg}), and $\|b_eb_e^\intercal\| = 2$, so using
$\ep \ge 1/n$\footnote{We can assume this WLOG since otherwise outputting $H=G$ suffices.} (hence $c\log n = 4\ep^{-1}\log n \le 4n\log n$) and
$|E|\le n^2$,
\[
\|W_r\| \le c\log n\cdot n\cdot 2\cdot |S| + \log n\cdot 2\cdot|E|
=O(n^5)
\]
and every entry of a matrix is bounded in magnitude by its norm.
Constructing $W_r$ explicitly is integer arithmetic on $O(\log n)$-bit
quantities and is hence in logspace. 
\end{proof}

\begin{proposition}\label{prop:phicomp}
There is a catalytic logspace algorithm that, given $G$, $\{q_e\}_{e\in E}$,
a state $\sigma = (S,\chi)$, and $N\le\poly(n)$ in unary, outputs $\widehat\Phi$ with $\poly(n)$ bits such that
$|\widehat\Phi-\Phi(\sigma)|\le 2^{-N}$.
\end{proposition}
\begin{proof}
Let $W_0, W_1$ be as in~\Cref{clm:Wnorm} and let
$A_r := L_G^{+/2}W_rL_G^{+/2}$, so that by~\Cref{def:potential} and
$\Tr(e^X) = \sum_{k\ge 0} \Tr(X^k)/k!$,
\[
\Phi(\sigma) = \sum_{k\ge 0} \frac{\Tr(A_0^k)+\Tr(A_1^k)}{k!} .
\]
By cyclicity of trace, $\Tr(A_r^k)=\Tr\big((L_G^+W_r)^k\big)$ for every
$k\ge 1$. Let $T:=O(n^7)$ be such that 
$\|L_G^+W_r\|, \|A_r\| \le T$, where the bound follows from~\Cref{clm:Wnorm} and the fact that $G$ is undirected (and hence $\|L_G^+\|=O(n^2)$).
 
\textbf{Truncation.} Set 
\[k' := 10\cdot (T + N + n).\]  
Using $k!\ge (k/e)^k$ and
$\Tr(A_r^k)\le n\|A_r\|^k \le nT^k$, for every $k\ge k'$ we have
$eT/k \le 1/2$ and hence
\[
\frac{|\Tr(A_r^k)|}{k!} \le n\left(\frac{eT}{k}\right)^{k}
\le n 2^{-k},
\]
so the total contribution of all terms with $k\ge k'$, over both values of
$r$, is at most $4n\cdot 2^{-k'} \le 2^{-N-2}$.
 
\textbf{Computing one term.}
Fix $k < k'$ and $r\in\zo$, and set
$N' := 20k'\log n + N=\poly(n)$.
By~\Cref{clm:Wnorm} the matrix $W_r$ can be constructed explicitly in
logspace, and by~\Cref{clm:compLperp} we can compute $\widehat{L}^{+}$ with
$\|\widehat{L}^{+}-L_G^{+}\|\le 2^{-N'}$ in $\CL$. Then
$\|\widehat{L}^{+}W_r - L_G^+W_r\| \le 2^{-N'}\cdot \|W_r\|$ and in particular
$\|\widehat{L}^{+}W_r\| \le T+1$, so by~\Cref{fct:trace-perturb},
\[
\Big|\Tr\big((\widehat{L}^{+}W_r)^k\big) - \Tr\big((L_G^+W_r)^k\big)\Big|
\le 2^{-N - 10k'\log n}=\gamma.
\]
Computing $\Tr\big((\widehat{L}^{+}W_r)^k\big)$ exactly is iterated
multiplication of $\poly(n)$ matrices with $\poly(n)$-bit rational entries, which is
computable by a logspace-uniform threshold circuit of depth $O(\log n)$ and
hence in $\CL$ by~\Cref{thm:BCKLS}; dividing by $k!$ and rounding the
result to $N+10k'\log n$ fractional bits incurs
additional error at most $2^{-N-10k'\log n}$ per term.
 
\textbf{Combining.} Summing the (at most $2k'$) rounded terms is iterated
addition which is computable in logspace. The total error
is at most
\[
\underbrace{2^{-N-2}}_{\text{truncation}}
+ \underbrace{2k'\cdot 2\cdot \gamma}_{\text{per-term}}
\le 2^{-N},
\]
and the output has $\poly(n)$ bits, as claimed.
\end{proof}

Our last ingredient is a lower bound on the potential function:
\begin{claim}\label{clm:PhiLB}
For every state $\sigma$ we have $\Phi(\sigma)\ge 2n$.
\end{claim}
\begin{proof}
    \[
    \Phi(\sigma) = \Tr\exp(cM_\sigma+V_{\sigma})+\Tr\exp(-cM_\sigma+V_{\sigma})
    \ge 2\Tr\exp(V_\sigma) \ge 2\Tr\exp(0) = 2n,
    \]
    where the first inequality is~\Cref{fct:TrExpconvex} and the second
    is~\Cref{fact:mono} together with $V_\sigma \succcurlyeq 0$.
\end{proof}

We can then prove the item.
\begin{proof}[Proof of~\Cref{itm:comp}]
Let $\Phi$ be the potential function of~\Cref{def:potential}, let
$A(\Phi,\sigma)$ denote the output of the algorithm of~\Cref{prop:phicomp}
on $(G,\{q_e\},\sigma)$ with $N := n^2+2$, and define
\[
\Phi'(\sigma) := A(\Phi,\sigma) + (|E|-|S|)\cdot 2^{-n^2},
\]
i.e.\ we add $2^{-n^2}$ for each edge not yet decided. Since a catalytic
machine computing a function must return the same output on every initial
catalytic tape, $A(\Phi,\sigma)$ --- and hence $\Phi'$ --- is a
well-defined function of $\sigma$. By~\Cref{prop:phicomp},
\begin{equation}\label{eq:abs-err}
\big|A(\Phi,\sigma) - \Phi(\sigma)\big| \le 2^{-n^2-2}
\qquad\text{for every state }\sigma,
\end{equation}
and in particular $\Phi'(\sigma) \ge \Phi(\sigma) - 1 > 0$
by~\Cref{clm:PhiLB}.
 
\textbf{$\Phi'$ satisfies~\Cref{itm:comp}} $A(\Phi,\sigma)$ and the slack term are computable in $\CL$, hence
$\Phi'$ is computable exactly in $\CL$ (and the output has bitlength $\poly(n)$).
 
\textbf{$\Phi'$ satisfies~\Cref{itm:super}.} Let $\sigma=(S,\chi)$ and
$e\notin S$; extending $\sigma$ decreases the slack term by exactly
$2^{-n^2}$. Then
\begin{align*}
\E_{b\sim\Bern(q_e)}\big[\Phi'(\sigma\cup(e\mapsto b))\big]
&= \E_b\big[A(\Phi,\sigma\cup(e\mapsto b))\big]
   + (|E|-|S|-1)\cdot 2^{-n^2}\\
&\le \E_b\big[\Phi(\sigma\cup(e\mapsto b))\big] + 2^{-n^2-2}
   + (|E|-|S|-1)\cdot 2^{-n^2} && \text{\Cref{eq:abs-err}}\\
&\le \Phi(\sigma) + 2^{-n^2-2} + (|E|-|S|-1)\cdot 2^{-n^2}
   && \text{\Cref{itm:super} for $\Phi$}\\
&\le A(\Phi,\sigma) + 2^{-n^2-1} + (|E|-|S|-1)\cdot 2^{-n^2}
   && \text{\Cref{eq:abs-err}}\\
&\le \Phi'(\sigma).
\end{align*}
 
\textbf{$\Phi'$ satisfies~\Cref{itm:impliesgood}}
Let $\sigma$ be a complete state with $\Phi'(\sigma)\le n^K\Phi'(\emptyset)$,
where $2^{K}\ep<1/2$. As $\sigma$ is complete its slack vanishes, so
by~\Cref{eq:abs-err},
\[
\Phi(\sigma) \le \Phi'(\sigma) + 2^{-n^2-2} \le n^K\Phi'(\emptyset)+1 ,
\qquad
\Phi'(\emptyset) \le \Phi(\emptyset) + 2^{-n^2-2} + n^2\cdot 2^{-n^2}
\le \Phi(\emptyset)+1 .
\]
Combining, and using $\Phi(\emptyset)\ge 2n$ (\Cref{clm:PhiLB}),
\[
\Phi(\sigma) \le n^K\Phi(\emptyset) + n^K + 1
\le n^K\Phi(\emptyset) + n^{K-1}\Phi(\emptyset)
\le n^{K+1}\Phi(\emptyset),
\]
so~\Cref{itm:impliesgood} for $\Phi$ applied with parameter $K+1$ shows $H(\sigma)$ is an
$\ep\cdot 2^{(K+1)/2}\le \ep\cdot 2^K$-spectral approximation of $G$ with
$O(n\ep^{-2}\log n)$ edges.

Finally, we let the potential function $\Phi$ of~\Cref{thm:pot} to be this function $\Phi'$, so we are done.
\end{proof}

\section{From the potential function to a prefix measure}\label{sec:measure}
Using the potential function, we define a measure on strings.
For a graph $G$, fix a global order $e_1,\ldots,e_m$ over the edges and a good set of probabilities $q_e$, and recall $q_e = 2^{-j_e}$ for $j_e\in \N$. For a string $x$ in
\[
\{0,1\}^{j_1}\times \ldots \{0,1\}^{j_m} =: \zo^B
\]
we define the state $\sigma_x$ in the natural way: we include edge $e$ if and only if the bits of $x$ in block $e$ are equal to $1^{j_e}$ (and if $j_e=0$ i.e. $e$ is a non-randomized edge, we always include the edge). Note that if we think of $x$ as a random string, this corresponds to including $e$ with probability $q_e$, exactly as in the leverage-score sparsifier.
\begin{theorem}\label{thm:measure}
    Fix $\ep\in (0,1/2)$, $K\in \N$ and a graph $G$. There is a function $f:\zo^*\ra\R$ defined on prefixes of $\zo^B$ with the following properties:
    \begin{enumerate}
        \item\label{itm:semi} $f$ is a semimeasure.
        \item\label{itm:comp2} Given $G$, $\eps$, and a prefix $x$, the value $f(x)$ is computable in $\CL$.
        \item\label{itm:impliesgood2} For a string $x\in \zo^B$, if $f(x) \le  n^K\cdot 2^{-|x|}$ then $\sigma_x$ defines a sparsifier $H(\sigma_x)$ with $O(n\ep^{-2}\log n)$ edges and $L_{H(\sigma_x)}\approx_{2^K \ep}L_G$.
    \end{enumerate}
\end{theorem}

\subsection{Defining the measure}
We first define the semimeasure, using the potential function $\Phi$ as defined in~\Cref{thm:pot}. We call the bits corresponding to $\{0,1\}^{j_i}$ block $i$. Note that every block is of length $O(\log n)$ and hence $|\chi|\le O(n^2\log n)$ by~\Cref{fct:pe-nonneg}. We define our measure as follows. First, assume the string $x$ is exactly an assignment to the first $k$ blocks. Let $\chi_x$ be the state where $\chi_e=1$ if and only if $x_{B_i}=1^{j_e}$. Then
\[
f(x) := 2^{-|x|}\frac{\Phi(\chi_x)}{\Phi(\emptyset)} = \prod_{i\in [k]}q_{e_i}\cdot \frac{\Phi(\chi_x)}{\Phi(\emptyset)}.
\]
For $x$ that consists of the first $k$ blocks and then $d<j_{k+1}$ bits of $B_{k+1}$, define
\[
f(x) = \sum_{s\in \zo^{j_{k+1}-d}}f(x||s).
\]
We encourage the reader to think of $f$ as defined over the blocks, with the bits being an implementation detail. We first show $f$ is in fact a semimeasure:
\begin{proof}[Proof of~\Cref{itm:semi}]
    Without essential loss of generality, assume that $x$ exactly fills the first $k$ blocks. Let $\chi^0$ and $\chi^1$ denote $\chi_x$ extended by setting $e_{k+1}$ to $0$ and $1$ respectively. For convenience, let $\gamma = \Phi(\emptyset)\cdot 2^{|x|}$ be a normalization factor. Since setting the first bit of block $k+1$ to $0$ ensures we will always set $\chi_{e_{k+1}}=0$, we have 
    \[\gamma f(x||0)=\sum_{s\in \{0,1\}^{j_{k+1}-1}}\gamma f(x||0||s) = 2^{j_{k+1}-1}\cdot 2^{-j_{k+1}}\Phi(\chi^0).\]
    Next, we have
    \[
    \gamma f(x||1)=\sum_{s}\gamma f(x||1||s) = 2^{-j_{k+1}}\Phi(\chi^1)+(2^{j_{k+1}-1}-1)2^{-j_{k+1}}\Phi(\chi^0)
    \]
    and hence
    \[
    \gamma f(x||0)+\gamma f(x||1) = q_{e_{k+1}} \Phi(\chi^1)+(1-q_{e_{k+1}})\Phi(\chi^0) \le \Phi(\chi_x) = \gamma f(x)
    \] 
    where the inequality follows from~\Cref{itm:super}, so we are done.
\end{proof}

Next, we show that given $G$ and $\{q_e\}_{e\in E}$, the semimeasure is computable in $\CL$.
\begin{proof}[Proof of~\Cref{itm:comp2}]
    Since $j_e=O(\log n)$, for a string $x$ that does not exactly fill the first $k$ blocks, we can enumerate over block suffixes $s$ using $O(\log n)$ space, so it suffices to compute $f$ on strings aligned to a block boundary. This follows immediately from $\Phi$ being computable in $\CL$ by~\Cref{itm:comp}.
\end{proof}

Finally, we show that any string that is not compressible according to $f$ must give a good sparsifier.
\begin{proof}[Proof of~\Cref{itm:impliesgood2}]
    Let $x$ be such a string where $f(x)\le n^K\cdot 2^{-|x|}$. This immediately implies $\Phi(\sigma_x)\le n^K \Phi(\emptyset)$. By~\Cref{itm:impliesgood} of~\Cref{thm:pot}, we have that $H$ satisfies the claimed properties.
\end{proof}

\section{In-place compression of a prefix measure}\label{sec:semicomp}

We show that a computable prefix semimeasure gives a $\CL$-computable compression scheme:
\semicomp*
For a string $x=x_1,\ldots,x_t$, we write $x_{\le i}$ for $x_1,\ldots,x_{i}$. For $y\in \zo^*$, we let $D_y$ be the dyadic interval represented by $y$, i.e. $[0.y,0.y+2^{-|y|})$.

We first define a compression scheme based on $f$, then show that we can compute it in-place. 
We use arithmetic coding~\cite{DBLP:journals/cacm/WittenNC87}, where we associate to each binary string a (unique, non-overlapping) interval in $[0,1]$, and encode a string $x$ by giving a dyadic interval that lies inside its interval. We define for every $x\in \zo^*$ an interval $I_x$, where $I_{\emptyset}=[0,1]$ and $I_{x||0}$ and $I_{x||1}$ are disjoint subintervals of $I_x$ of length $f(x||0)$ and $f(x||1)$ respectively. We adopt the convention that $I_{x||0}$ is the leftmost length-$f(x||0)$ subinterval of
$I_x$ and $I_{x||1}$ the length-$f(x||1)$ subinterval immediately to its right, so that
any slack $f(x)-f(x||0)-f(x||1)\ge 0$ sits at the right end of $I_x$.
\begin{claim}\label{clm:compI}
    There exists a $\CL$ algorithm that, given $x$, outputs (the endpoints of) $I_x$. 
\end{claim}
\begin{proof}
Write $x_{<i}:=x_1\cdots x_{i-1}$. By our convention on defining the intervals, passing
from $I_w$ to $I_{w||1}$ shifts the left endpoint right by exactly $f(w||0)$, while passing
to $I_{w||0}$ leaves it fixed; hence letting $a_x$ be the left endpoint of $I_{x}$ we have
\[
  a_x = \sum_{i:x_i=1} f(x_{<i}||0).
\]
The right endpoint is then $a_x+f(x)$ which we can calculate in $\CL$ by assumption. Finally, the left endpoint is a sum of at most $t$ numbers $a_i=f(x_{<i}||0)$ of bitlength $t^c$ that are computable in $\CL$ and the map $(a_1,\ldots,a_t)\ra a_1+\ldots+a_t$ is computable in logspace. Thus we have that $a_x$ is computable in $\CL$ by composition. 
\end{proof}

The arithmetic encoding $A(x)$ is the largest dyadic interval contained in $I_x$, which we specify with a string $y$ where $D_y\subseteq I_x$. It is easy to see that such a $y$ exists with $|y|\le -\log f(x)+2$. The main component of the proof is showing that we can perform encoding and decoding in-place. 

\subsection{The block structure of the encoding}
Fix a string $x=x_1,\ldots,x_t$, and for $0\le i\le t$ write
\[
\lambda(i):=-\log_2 f(x_{\le i}),
\qquad
s(i):=i-\lambda(i).
\]
We assume without loss of generality that $f(x_{\le i})>0$ for every prefix under consideration. Intuitively, $s(i)$ is measuring how compressible $x_1 x_2 \dots x_i$ is, where $t$ measures the original length and $\lambda(i)$ represents (approximately) the size of the compressed representation of $x_1 x_2 \dots x_i$. Thus, $s(i)$ represents the number of bits freed up by compressing $x_1 x_2 \dots x_i$.

We first describe an encoder that outputs $A(x)$ given $x$ but does not transform $x$ into $A(x)$ in-place, then show how to implement it in-place (and how to implement the decoder). For $x_{\le i}$, let $p_i$ be the longest binary string where $D_{p_i}\supseteq I_{x_{\le i}}$,
and let
$\mu_{p_i}$ denote the midpoint of $D_{p_i}$. For $S \ge 0$, let
\[
D_{p_i}^{(S)} := \big[\mu_{p_i} - 2^{-|p_i|-S-1},\ \mu_{p_i} + 2^{-|p_i|-S-1}\big)
\]
be the interval of length $2^{-S}\cdot|D_{p_i}|$ centered at $\mu_{p_i}$,
so that $D_{p_i}^{(0)} = D_{p_i}$. We define $S(i)$ to be the largest
integer such that $D_{p_i}^{(S(i))} \supseteq I_{x_{\le i}}$.

We think of the output of the encoder on $x_1,\ldots,x_i$ as
\[
\oU_{\le i}:= \begin{cases}
    p_i & S(i)=0\\
    p_i||u||\bar{u}^{S(i)-1} & S(i)>0
\end{cases}
\]
where $u$ is a special undetermined symbol will be fixed to $\sigma \in \zo$ by future output, and we define $\oA_{\le i}$ to be the string with this value fixed. In both cases $|\oU_{\le i}| = |p_i| + S(i)$.

Note that $S(i)>0$ if the current interval $I_{x_{\le i}}$ is small relative to the containing interval $D_{p_i}$, but still contains its midpoint $\mu_{p_i}$ of $D_{p_i}$.
\begin{claim}
    There is a $\CL$ algorithm that, given $x_{\le i}$, outputs $p_i$ and $S(i)$.
\end{claim}
\begin{proof}
    Since we can compute $I_{x_{\le i}}$ in $\CL$ by~\Cref{clm:compI}, the rest follows because arithmetic is in logspace.
\end{proof}
We prove a basic bound on the length of $E_{\le i}$:
\begin{lemma}\label{lem:Elength}
For every $i$ we have $|\oU_{\le i}|\in\bigl(\lambda(i)-2,\ \lambda(i)\bigr]$.
\end{lemma}
\begin{proof}
    Write $p := p_i$, $S := S(i)$, $\mu := \mu_{p_i}$, and
    $\lambda := \lambda(i)$; recall $|\oU_{\le i}| = |p| + S$. By maximality
    of $p$, the interval $I_{x_{\le i}}$ is contained in neither dyadic
    child of $D_p$, so we may write
    \[
    I_{x_{\le i}} = [\mu - \alpha,\ \mu+\beta), \qquad \alpha,\beta > 0,
    \qquad \alpha + \beta = |I_{x_{\le i}}| = 2^{-\lambda}.
    \]
    Set $M := \max(\alpha,\beta)$, and note
    \begin{equation}\label{eq:M-sandwich}
    2^{-\lambda-1} \le M \le 2^{-\lambda},
    \end{equation}
    the lower bound since $M \ge (\alpha+\beta)/2$ and the upper since
    $M \le \alpha+\beta$.
 
    Then note that the containment
    $I_{x_{\le i}} \subseteq D_p^{(S')}$ holds if and only if both $\alpha$
    and $\beta$ are at most $2^{-|p|-S'-1}$, i.e.\ if and only if
    \[
    M \le 2^{-|p|-S'-1} .
    \]
    Since $S$ is the largest such $S'$, we obtain
    \begin{equation}\label{eq:S-sandwich}
    2^{-|p|-S-2} < M \le 2^{-|p|-S-1}.
    \end{equation}
    and combining these bounds with the bounds of~\Cref{eq:M-sandwich} gives the result. 
\end{proof}

We think of the encoder as printing $|\oU_{\le i}|$ symbols after processing $x_{\le i}$ (despite the fact that some symbols may not yet be determined). We show that this is well-defined by showing that subsequent output always extends these prefixes in a consistent way.
\begin{fact}[Midpoint chains]\label{fct:midpoint-chains}
    For every $p\in\zo^*$ and $k\ge 1$,
    \[
    D_{p||1||0^{k-1}} = \big[\mu_p,\ \mu_p+2^{-|p|-k}\big)
    \qquad\text{and}\qquad
    D_{p||0||1^{k-1}} = \big[\mu_p-2^{-|p|-k},\ \mu_p\big).
    \]
\end{fact}

\begin{claim}\label{clm:extend}
For every $i$, either $p_{i+1}=p_i$, or $p_{i+1}$ is a proper extension of
\[
p_i||\sigma||\bar\sigma^{S(i)-1}
\]
for some $\sigma\in\zo$ (where for $S(i)=0$ this string is read as $p_i$).
In the latter case, $\sigma$ is the unique bit with
$I_{x_{\le i+1}}\subseteq D_{p_i||\sigma}$.
\end{claim}
\begin{proof}
    Write $p:=p_i$, $p'=p_{i+1}$, $S:=S(i)$, $\mu:=\mu_{p}$, and
    $I':=I_{x_{\le i+1}}$, so that $I'\subseteq I_{x_{\le i}}\subseteq
    D_{p}^{(S)}$.
     
    \textbf{Step 1: $p'$ extends $p$.}
    Both $D_{p'}$ and $D_{p}$ contain the nonempty interval $I'$, and any
    two dyadic intervals are either nested or disjoint; hence they are nested.
    If $D_{p}\subsetneq D_{p'}$ then $|p'|<|p|$, contradicting that
    $p'$ is the longest string whose interval contains $I'$ (as
    $D_p\supseteq I'$). So $D_{p'}\subseteq D_{p}$, i.e.\ $p_{i+1}$
    extends $p$. If $p'=p$ we are in the first case, so
    assume $p_{i+1}$ extends $p$ by $j\ge 1$ bits, and let $\sigma$ be the
    first extension bit. Then $I'\subseteq D_{p'}\subseteq D_{p||\sigma}$,
    and $\sigma$ is unique since the two children of $D_p$ are disjoint.
     
    \textbf{Step 2: for every $1\le k\le S+1$ we have
    $I'\subseteq D_{p||\sigma||\bar\sigma^{k-1}}$.}
    We give the case $\sigma=1$; the case $\sigma=0$ is symmetric. Since
    $I'\subseteq D_{p||1}$, every point of $I'$ is $\ge\mu$. Since
    $I'\subseteq D_{p}^{(S)}$, every point of $I'$ is
    $<\mu+2^{-|p|-S-1}\le \mu+2^{-|p|-k}$ for $k\le S+1$. Hence
    $I'\subseteq[\mu,\ \mu+2^{-|p|-k}) = D_{p||1||0^{k-1}}$
    by~\Cref{fct:midpoint-chains}. It is thus straightforward to show that $p'$ extends $p||1||0^{S}$ (and hence is a proper extension of $p||1||0^{S-1}$), since otherwise there would be a longer $p''$ whose interval contains $I'$.
\end{proof}
Note that the latter case occurs if $I_{x_{\le i+1}}$ no longer contains the midpoint $\mu_{p_i}$ of $D_{p_i}$. Once this occurs, we think of the encoder as going back and replacing the previously outputted string $u||\bar{u}^{S(i)-1}$ with $\sigma||\bar{\sigma}^{S(i)-1}$ for $\sigma\in \zo$ (and then outputting any further bits of $p_{i+j}$). We formally define this as the idea of a resolving bit:
\begin{definition}[Resolution]\label{def:resolution}
  The resolving bit of $\oU_{\le i}$ (when
  $S(i)\ge 1$) is, letting $j>i$ be the first index where $p_j$ strictly extends $p_i$, the first bit of this extension. If no such $j$ exists, it is the bit $\sigma$ of
  \Cref{lem:final-block}. We write
  $\oA_{\le i}$ for $\oU_{\le i}$ with $u$ resolved in this way.
\end{definition}
\begin{corollary}\label{cor:R-prefix}
    If $p_{i+1}=p_i$ then $S(i+1)\ge S(i)$, and $u$ resolves to the same bit $\sigma$ in both blocks. Consequently, $\oA_{\le i}$ is a
    (not necessarily proper) prefix of $\oA_{\le i+1}$ for every $i$, and we can define
    \[
    \oA_i := \oA_{\le i}-\oA_{\le i-1}\in \zo^*
    \]
\end{corollary}

Finally, once we finish processing $x_t$, the current output $\oU_{\le t}=p_t||u||\bar{u}^{S(t)-1}$
still has $S(t)$ pending symbols undetermined, and the interval $D_{p_t}$ contains $I_{x}$, rather than lying inside it. We show that the final encoding extends $\oU_{\le t}$ in the natural way:
\begin{lemma}[Termination]\label{lem:final-block}
Write $I_x = [\mu_{p_t}-\alpha,\ \mu_{p_t}+\beta)$ and let $\sigma := 1$
if $\beta\ge\alpha$ and $\sigma:=0$ otherwise. Then
$y := p_t||\sigma||\bar\sigma^{S(t)+1}$ satisfies $D_y\subseteq I_x$
and $|y| = |p_t|+S(t)+2$.
\end{lemma}
\begin{proof}
We give the case $\sigma=1$; the case $\sigma=0$ is symmetric. By the
choice of $\sigma$ and~\Cref{eq:S-sandwich},
$\beta=\max(\alpha,\beta)>2^{-|p_t|-S(t)-2}$. By~\Cref{fct:midpoint-chains}
with $k=S(t)+2$,
\[
D_y = \big[\mu_{p_t},\ \mu_{p_t}+2^{-|p_t|-S(t)-2}\big)
\subseteq \big[\mu_{p_t},\ \mu_{p_t}+\beta\big)\subseteq I_x. \qedhere
\]
\end{proof}
We let $\oA_f:=y-\oA_{\le t}$, so the final encoding is 
\[\oA(x) := \oA_{\le t}||\oA_{f}.\]

\subsubsection{In-place encoding and decoding}
The reverse encoder for a string $x\in \zo^t$ works as follows. For $i=t,t-1,\ldots,1$ it replaces $x_{i+1:t}$\footnote{We actually add a constant amount of padding to ensure there is always sufficient space} with $\oA_{i+1,\ldots,t}||\oA_{f}$, while always storing the last committed bit $\sigma$. Since we can compute $\oU_{\le i}$ using only $x_{1:i}$ and the bit $\sigma$, after losing access to $x_{i+1,\ldots,t}$ we retain the ability to compute $\oA_{\le i}$. The final step is to show that for each $i$, the number of bits we have output does not exceed (by more than a constant) the number of bits we have freed up. This is \textit{not} true for a generic $x$, but it does occur if $x$ has a particular property. 
\begin{definition}
    We say a string $x\in \zo^t$ has a \emphdef{final compression record (FCR)} if $s(t)\ge s(i)$ for every $i\le t$.
\end{definition}
\begin{proposition}\label{prop:records}
Suppose $x\in\zo^t$ has an FCR. Then for every
$0\le i\le t$,
\[
\big|\oA_{i+1,\ldots,t}||\oA_f\big| \le t-i+4.
\]
\end{proposition}
\begin{proof}
By \Cref{lem:final-block} and \Cref{lem:Elength},
\[
\big|\oA_{i+1,\ldots,t}||\oA_f\big| = |\oA(x)|-|\oA_{\le i}|
< \big(\lambda(t)+2\big)-\big(\lambda(i)-2\big)
= (t-i)-\big(s(t)-s(i)\big)+4 \le t-i+4,
\]
using $s(t)\ge s(i)$.
\end{proof}

We then give the encoder algorithm, where we assume that $x$ has an FCR and $-\log f(x)<t-C\log t$.
\begin{remark}[Subroutine convention]\label{rem:subroutine}
    Throughout, ``compute $Z$ from $x_{\le i}$ in $\CL$'' means that we run the
    corresponding $\CL$ algorithm with read--only access to the prefix of the
    primary tape holding $x_{\le i}$, using the catalytic tape of the
    $\inplaceFCL$ model as the catalytic tape for the subroutine (and since the subroutine resets the tape, this does not affect the invoking algorithm).
\end{remark}

\newcommand{\Buf}{\mathrm{Buf}}
\newcommand{\Inv}{\mathrm{Inv}}

\begin{algorithm}
\caption{\textsc{ReverseEncode}: in-place arithmetic encoding}\label{alg:reverse-encode}
\begin{algorithmic}[1]
\Require primary tape holds $x\in\zo^t$ with a final compression record
\Ensure primary tape holds $0^{g}||\oA_{\le t}||\oA_{f}=0^g||\oA(x)$, where $g = t-|\oA_{\le t}|-2$
\State compute $p_t, S(t)$ and the final block: $\sigma^{\mathrm{cur}}\gets$ heavier-side bit, $F\in\zo^2$ \Comment{\Cref{lem:final-block}}
\State $\Buf\gets F$;\quad $w\gets 2$ \Comment{$\Buf$: the $\le 5$ leftmost code bits; $w$: total bits emitted}
\For{$i = t$ \textbf{downto} $1$}
    \State compute $p_i, S(i), p_{i-1}, S(i-1)$ from $x_{\le i}$ in $\CL$
    \If{$p_i \ne p_{i-1}$} \Comment{prefix jump; \Cref{clm:extend} applies at index $i-1$}
        \State $\sigma^{\mathrm{prev}}\gets$ first bit of $p_i$ after the prefix $p_{i-1}$
    \Else
        \State $\sigma^{\mathrm{prev}}\gets \sigma^{\mathrm{cur}}$
    \EndIf
    \State \parbox[t]{0.9\linewidth}{$\oA_i\gets$ the suffix of
    $p_i||\sigma^{\mathrm{cur}}||\overline{\sigma^{\mathrm{cur}}}^{S(i)-1}$
    obtained by removing the prefix
    $p_{i-1}||\sigma^{\mathrm{prev}}||\overline{\sigma^{\mathrm{prev}}}^{S(i-1)-1}$
    \Comment{well defined by \Cref{cor:R-prefix}}}
    \State prepend $\oA_i$ to the code block, bit by bit: new leftmost bits enter $\Buf$, and bits displaced from $\Buf$ spill onto the tape immediately left of the on-tape code block
    \State $w\gets w+|\oA_i|$;\quad $\sigma^{\mathrm{cur}}\gets\sigma^{\mathrm{prev}}$
\EndFor
\State flush $\Buf$ onto the tape and write $0^{t-w}$ on the prefix
\end{algorithmic}
\end{algorithm}

\begin{lemma}[Correctness of \Cref{alg:reverse-encode}]\label{lem:reverse-encode}
\Cref{alg:reverse-encode} is an $\inplaceFCL$ algorithm satisfying its
specification.
\end{lemma}
\begin{proof}
Let $\Inv(i)$ be the following claim. \emph{positions $1,\ldots,i$ of the main tape
hold $x_{\le i}$; the string $\oA_{i+1,\ldots,t}||\oA_{f}$ (of length $w$) is
stored with its leftmost $\min(w,5)$ bits in $\Buf$ and the
remainder right-aligned on the main tape; and $\sigma^{\mathrm{cur}}$ is the
resolving bit of $\oU_{\le i}$ (and is unused if $S(i)=0$).} It is immediate that $\Inv(t)$ holds at the start of the algorithm, and that $\Inv(0)$ implies correctness. We show that assuming $\Inv(i)$ holds at the start of the loop for $i$, then $\Inv(i-1)$ holds at the end of this loop execution.

\begin{itemize}
    \item By \Cref{cor:R-prefix}, $\oA_{\le i-1}$ is a prefix of $\oA_{\le i}$, and its resolving bit is equal to $\sigma^{\mathrm{cur}}$ itself when $p_i=p_{i-1}$, and otherwise equal to the first bit of $p_i$ not contained in $p_{i-1}$ (\Cref{clm:extend}), which is itself equal to $\sigma^{\mathrm{prev}}$. Hence line~9
    computes $\oA_i = \oA_{\le i}-\oA_{\le i-1}$, and prepending it extends
    $\oA_{i+1,\ldots,t}||\oA_{f}$ to $\oA_{i,\ldots,t}||\oA_{f}$ and hence $\Inv_2(i-1)$ holds.
    \item It is easy to see that $\sigma^{\mathrm{cur}}$ is updated to the
    level-$(i-1)$ resolving bit so $\Inv_3(i-1)$ holds.
    \item At the end of the iteration, the code block
    $\oA_{i,\ldots,t}||\oA_f$ has length at most $t-i+5$
    by~\Cref{prop:records} applied at index $i-1$. Since $\Buf$
    holds its $5$ leftmost bits, the on-tape portion occupies positions
    $\ge t-(t-i)+1 = i+1$, and hence $x_{\le i-1}$ is unmodified into the next iteration so $\Inv_1(i-1)$ holds.
\end{itemize}
Finally, it is clear that the algorithm can be implemented in-place, storing $i$, $w$, $\sigma^{\mathrm{cur}}$,
$\Buf$ and pointers on the worktape using $O(\log t)$ space.
\end{proof}

The decoder algorithm works similarly, except that we iteratively replace $\oA_{\le i}$ with $x_{\le i}$. We first show that in $\CL$ we can decode the next symbol given the decoded previous symbols and the remainder of the encoding.
\begin{claim}\label{clm:dec-map}
There is a $\CL$ algorithm that, given $x_{<i}$, the bit $\sigma$
resolving the pending block of $\oU_{\le i-1}$ (when $S(i-1)\ge 1$), and
the string $\oA_{i,\ldots,t}||\oA_f$, outputs $x_i$.
\end{claim}
\begin{proof}
    From $x_{<i}$ we compute $p_{i-1}$ and $S(i-1)$ in $\CL$, and hence, with
    $\sigma$, the resolved prefix
    $\oA_{\le i-1} = p_{i-1}||\sigma||\bar\sigma^{S(i-1)-1}$; concatenating
    the given suffix yields bit access to the full code $\oA(x)$. It is straightforward to show that
    $x_i$ is the unique bit $b$ with $D_{\oA(x)} \subseteq I_{x_{<i}||b}$. Since we can compute these intervals in $\CL$ by~\Cref{clm:compI} and we can determine $D_{\oA(x)}$ from $\oA(x)$ in logspace, we are done.
\end{proof}

We now give the implementation of the decoder, which works in the opposite direction: for $i=1,\ldots,t$ we replace the $i$-bit prefix of the tape with $x_{\le i}$. We again show that if $x$ has an FCR, such a process can be implemented with a constant sized buffer. 
\begin{algorithm}
\caption{\textsc{ReverseDecode}: in-place arithmetic decoding}\label{alg:reverse-decode}
\begin{algorithmic}[1]
\Require primary tape holds $0^{g}||\oA(x)$, the output of \Cref{alg:reverse-encode} on some $x\in\zo^t$ with a final compression record
\Ensure primary tape holds $x$
\State $\sigma^{\mathrm{cur}}\gets\bot$;\quad $\Buf\gets\varepsilon$ \Comment{$\Buf$: live code bits about to be overwritten, $\le 5$}
\For{$i = 1$ \textbf{to} $t$}
    \State compute $x_i$ from $x_{<i}$, $\sigma^{\mathrm{cur}}$, and the code (on tape / in $\Buf$ / reconstructed) via \Cref{clm:dec-map}
    \State move any still-live code bit at position $i$ into $\Buf$; write $x_i$ at position $i$
    \State compute $p_i, S(i), p_{i-1}, S(i-1)$ from $x_{\le i}$ in $\CL$
\If{$S(i)\ge 1$ \textbf{and} ($p_i \ne p_{i-1}$ \textbf{or} $S(i-1)=0$)}
    \State $\sigma^{\mathrm{cur}}\gets$ code bit at index $|p_i|+1$, read from the tape or $\Buf$
\ElsIf{$S(i)=0$}
    \State $\sigma^{\mathrm{cur}}\gets\bot$
\EndIf \Comment{otherwise $p_i=p_{i-1}$, $S(i-1)\ge1$ hence $\sigma^{\mathrm{cur}}$ unchanged}
\State discard from $\Buf$ any bit whose code index is now $\le |\oA_{\le i}|$ \Comment{after the read above}
\EndFor
\end{algorithmic}
\end{algorithm}
 
\begin{lemma}[Correctness of \Cref{alg:reverse-decode}]\label{lem:reverse-decode}
\Cref{alg:reverse-decode} is an $\inplaceFCL$ algorithm satisfying its
specification.
\end{lemma}
\begin{proof}
Index the bits of the code $\oA(x)$ by $1,\ldots,|\oA(x)|$, so that bit
$j$ sits at tape position $g+j$ initially. Let $\Inv(i)$ be the following claim: \emph{positions $1,\ldots,i-1$ hold $x_{<i}$;
every code bit with index $j > |\oA_{\le i-1}|$ is available, either on
the tape at position $g+j$ (if $g+j\ge i$) or in $\Buf$; and
$\sigma^{\mathrm{cur}}$ resolves the pending block of $\oU_{\le i-1}$
(unused if $S(i-1)=0$).} It is immediate that $\Inv(1)$ holds at the start of the algorithm, and that $\Inv(t+1)$ implies correctness. We show that assuming $\Inv(i)$ holds at the start of the loop for $i$, then $\Inv(i+1)$ holds at the end of this loop execution.

\begin{itemize}
    \item line~3 correctly computes $x_i$ by~\Cref{clm:dec-map} and hence $\Inv_2(i+1)$ holds by Line 4.
    \item For the $\sigma^{\mathrm{cur}}$ update, there are three cases. If
    $p_i=p_{i-1}$ and $S(i-1)\ge 1$, then the pending block of $\oU_{\le i}$
    is the same block as that of $\oU_{\le i-1}$ (\Cref{cor:R-prefix}), so
    keeping $\sigma^{\mathrm{cur}}$ is correct. If $S(i)=0$ there is no
    pending block and $\sigma^{\mathrm{cur}}$ is unused. Otherwise the resolving bit of $\oU_{\le i}$ is equal to bit
    $|p_i|+1$ of the code (\Cref{def:resolution}). We claim this bit has not yet been erased. If
    $p_i\neq p_{i-1}$ then $p_i$ properly extends $p_{i-1}||\sigma||\bar\sigma^{S(i-1)-1}$
    (\Cref{clm:extend}), so $|p_i|+1 \ge |\oA_{\le i-1}|+2$, and otherwise $S(i-1)=0$ so $|p_i|+1 = |\oA_{\le i-1}|+1$. Thus in both cases the index exceeds $|\oA_{\le i-1}|$, so by $\Inv(i)$
    the bit is on the tape or in $\Buf$, and it is read before the
    discard step, so $\Inv_3(i+1)$ holds. 
    \item  The bits held in $\Buf$ after iteration $i$
    are those with code index $j$ satisfying $g+j \le i$ and
    $j > |\oA_{\le i}|$, so
    \[
    |\Buf| \le (i-g) - |\oA_{\le i}|
    < i - g - \lambda(i)+2
    = s(i) - g + 2
    \le s(t) - g + 2
    \le 4,
    \]
    using \Cref{lem:Elength}, the final compression record, and
    $g = t-|\oA(x)| \ge t-\lambda(t)-2 = s(t)-2$, and hence $\Inv_1(i+1)$ holds. 
\end{itemize}
Finally, it is clear that the algorithm can be implemented in-place storing $i$, $\sigma^{\mathrm{cur}}$, $\Buf$ and pointers on the worktape using $O(\log t)$ space.
\end{proof}

\subsection{Putting it all together}
Recall that we are given $x\in \zo^t$ where $-\log f(x)<t-2C\log t$, but $x$ does not necessarily have an FCR. However, it is straightforward to deal with this: there is a \textit{prefix} $x'=x_{\le t'}$ of $x$ that is a compression record and $-\log f(x')<t'-2C\log t'$. Our final algorithm $\Enc$ simply encodes this prefix using~\Cref{alg:reverse-encode}, then preserves the suffix in the clear. We give the encoding and decoding algorithm below, then the proof of~\Cref{thm:semicomp} simply verifies that they obey their specification:
\begin{algorithm}
\caption{$\Enc$: the full in-place encoder}\label{alg:enc}
\begin{algorithmic}[1]
\Require primary tape holds $x\in\zo^t$ with $-\log f(x)<t-2C\log t$
\Ensure primary tape holds $\oE(x)||0^{C\log t}$
\State $t'\gets$ the least index with $s(t')\ge (C+3)\log t$ \Comment{exists since $s(t)>2C\log t$}
\State run \Cref{alg:reverse-encode} on the tape region $[1,t']$, leaving $0^{g'}||\oA(x_{\le t'})||x_{t'+1,\ldots,t}$
\State $\ell\gets |\oA(x_{\le t'})|$
\State shift the code $\oA(x_{\le t'})$ left so that it occupies positions $h+1,\ldots,h+\ell$, and write $\langle t'\rangle||\langle\ell\rangle$ in positions $1,\ldots,h$
\State shift the suffix $x_{t'+1,\ldots,t}$ left so that it occupies positions $h+\ell+1,\ldots,h+\ell+(t-t')$, and write $0$'s on the vacated cells
\end{algorithmic}
\end{algorithm}

\begin{algorithm}
\caption{$\Dec$: the full in-place decoder}\label{alg:dec}
\begin{algorithmic}[1]
\Require primary tape holds $\oE(x)||0^{C\log t}$
\Ensure primary tape holds $x$
\State read $t'$ and $\ell$ from the header in positions $1,\ldots,h$
\State shift the raw suffix (positions $h+\ell+1,\ldots,h+\ell+(t-t')$) right so that it occupies positions $t'+1,\ldots,t$ \Comment{copy in decreasing order}
\State shift the code (positions $h+1,\ldots,h+\ell$) right so that it occupies positions $t'-\ell+1,\ldots,t'$, and write $0$'s on positions $1,\ldots,t'-\ell$
\State run \Cref{alg:reverse-decode} on the tape region $[1,t']$
\end{algorithmic}
\end{algorithm}

\begin{proof}[Proof of~\Cref{thm:semicomp}]~
Let $h=2\log t$ for clarity.
\paragraph{Correctness of the encoder}
The index $t'$ exists since
$s(t)>2C\log t \ge (C+3)\log t$, so \Cref{alg:reverse-encode}
can be validly invoked on the prefix. For the layout,
\[
t'-\ell \ge t'-\lambda(t')-2 = s(t')-2
\ge C\log t + h,
\]
so after lines~4--5 the tape holds
\[
\oE(x)||0^{C\log t},
\qquad
\oE(x) := \langle t'\rangle||\langle\ell\rangle||\oA(x_{\le t'})||
x_{t'+1,\ldots,t}||0^{t'-\ell-h-C\log t}
 \in \zo^{t-C\log t}.
\]
Both shifts move contiguous
blocks left into free space and are performed by copying in increasing
position order with $O(\log n)$ counters.

\paragraph{Correctness of the decoder}
The header determines the boundaries of all three
blocks exactly, so lines~2--3 invert the two shifts of \Cref{alg:enc}. After
line~3 the region $[1,t']$ holds $0^{g'}||\oA(x_{\le t'})$, exactly the
postcondition of \Cref{alg:reverse-encode} on $x_{\le t'}$, so line~4
restores $x_{\le t'}$ by \Cref{lem:reverse-decode}, and the tape holds
$x_{\le t'}||x_{t'+1,\ldots,t} = x$.
\end{proof}

\section{Putting it all together}\label{sec:final}

\begin{proof}[Proof of~\Cref{thm:main}]
    Let $G=(V,E)$ be the graph on $n$ vertices given as input and let $m=|E|$. Let $\{q_e\}_{e\in E}$ be the set of good probabilities that exist and can be computed in $\CL$ (\Cref{cor:CLqe}), and recall that $q_e = 2^{-j_e}$. We think of our catalytic tape as $\tau_1,\ldots,\tau_B,\tau'$ where
    \[
    \tau_i \in \zo^{j_1}\times\ldots\times \zo^{j_m}
    \]
    where if $j_e=0$ we do not allocate any space on the tape (corresponding to edges that are always included),
    and $\tau'$ is used as additional catalytic space for subroutines as necessary. 
    Note that $B=\sum_e j_e\le n^3$.
    Moreover, let $f$ be the function of~\Cref{thm:measure} with $K=18$ and error parameter $\eps/2^K$. The algorithm first tests if for any $\tau_i$ we have $f(\tau_i)\le n^K2^{-|\tau_i|}$ (using that $f$ is computable in $\CL$). If this holds, we output the sparsifier $H(\chi_{\tau_i})$ (where we include edge $e$ if the relevant block is equal to $1^{j_e}$), and by Item 3 we have that this graph has $O(n\eps^{-2}\log n)$ edges and $L_{H(\tau_i)}\approx_\eps L_G$. Otherwise, we must have
    \[
    -\log f(\tau_i)< B-\log(n^{18})\le B-6\log B, 
    \]
    so we invoke the algorithm $\Enc$ of~\Cref{thm:semicomp} with $C=3$ on $\tau_i$ for every $i$. This means that the $i$th block of the tape now has the final $0^{3\log B}$ bits free, and hence there are $B3\log B\ge B$ bits free on the tape.\footnote{We can assume WLOG that $3\log B\ge 1$, since otherwise we can brute force over the $B=O(1)$ bits corresponding to randomized edges using the worktape.} Thus, we use this free space to brute force over $\chi\in \zo^B$ to find a $\chi$ where $f(\chi)\le n^K2^{-|\chi|}$ (which exists by an averaging argument), then output the sparsifier $H(\chi)$. Finally, we run the algorithm $\Dec$ of~\Cref{thm:semicomp} on $\tau_i$ for every $i$, which resets the tape.
\end{proof}

\section*{Statement on AI Use}
Compressing high-potential strings using in-place arithmetic coding (both the idea and initial writeup) was developed jointly with Fable 5. The authors take full responsibility for the content of the document.

\section*{Acknowledgments}
Edward Pyne thanks Dean Doron and Aryan Agarwala for conversations about sparsification in $\CL$. This material is based upon work supported by the Air Force Office of Scientific Research under award number FA9550-23-F-0014 in the amount of \$139,400 to Cassandra Marcussen. Cassandra Marcussen is also supported in part by NSF Award 2152413 and a
Simons Investigator Award to Madhu Sudan. Edward Pyne is supported by an NSF Graduate Research Fellowship. Ronitt Rubinfeld is supported by the NSF TRIPODS program (award DMS-2022448) and CCF-2310818.

\bibliographystyle{alpha}
\bibliography{ref}

\end{document}

%% file: 0-macros.tex
\usepackage{amsmath,amsfonts,amssymb}
\usepackage{amsthm}
\usepackage{bbm}
\usepackage{bm}
\usepackage{latexsym}
\usepackage{tikz,circuitikz}
\usepackage[noadjust]{cite}
\usepackage{graphicx}
\usepackage{sepfootnotes}

\usepackage[colorlinks=true,linkcolor=blue,citecolor=blue]{hyperref}
\usepackage{xcolor}
\usepackage{dirtytalk}
\usepackage{mathtools}
\usepackage[margin=1in]{geometry}
\usepackage{thm-restate}
\usepackage{enumitem}
\usepackage{tikz-cd}
\usepackage{algorithm}
\usepackage{algpseudocode}
\usepackage{comment}
\newtheorem{theorem-intro}{Theorem}
\newtheorem{corollary-intro}{Corollary}
\newtheorem{theorem}{Theorem}[section]
\newtheorem{corollary}[theorem]{Corollary}

\newtheorem{lemma}[theorem]{Lemma}

\newtheorem{proposition}[theorem]{Proposition}
\newtheorem{claim}[theorem]{Claim}
\theoremstyle{definition}
\newtheorem{definition}[theorem]{Definition}
\newtheorem{remark}[theorem]{Remark}

\newtheorem{fact}[theorem]{Fact}

\usepackage[nameinlink]{cleveref}
\Crefname{assumption}{Assumption}{Assumptions}
\Crefname{openprob}{Open Problem}{Open Problems}
\Crefname{observation}{Observation}{Observations}
\Crefname{claim}{Claim}{Claims}
\Crefname{theorem}{Theorem}{Theorems}
\Crefname{corollary-intro}{Corollary}{Corollaries}
\Crefname{theorem-intro}{Theorem}{Theorems}
\Crefname{problem}{Problem}{Problem}
\Crefname{proposition}{Proposition}{Propositions}
\Crefname{definition}{Definition}{Definition}
\Crefname{lemma}{Lemma}{Lemma}
\Crefname{conjecture}{Conjecture}{Conjecture}
\Crefname{fact}{Fact}{Fact}
\Crefname{corollary}{Corollary}{Corollaries}
\Crefname{ineq}{inequality}{inequalities}
\Crefname{equation}{Equation}{Equations}

\newcommand{\N}{\mathbb{N}}

\newcommand{\R}{\mathbb{R}}

\newcommand{\one}{\mathbf{1}}
\newcommand{\ep}{\eps}
\newcommand{\ra}{\rightarrow}

\newcommand{\zo}{\{0,1\}}
\newcommand{\eps}{\varepsilon}

\newcommand{\tO}{\widetilde{O}}

\DeclareMathOperator{\polylog}{polylog}
\DeclareMathOperator{\poly}{poly}
\DeclareMathOperator*{\E}{\mathbb{E}}

\newcommand{\SC}{\bm{\mathsf{SC}}}

\newcommand{\NL}{\bm{\mathsf{NL}}}
\newcommand{\DET}{\bm{\mathsf{DET}}}

\newcommand{\CL}{\bm{\mathsf{CL}}}

\newcommand{\inplaceFCL}{\bm{\mathsf{inplaceFCL}}}
\newcommand{\searchCL}{\bm{\mathsf{searchCL}}}
\newcommand{\LOSSY}{\bm{\mathsf{LOSSY}}}

\renewcommand{\P}{\bm{\mathsf{P}}}

\newcommand{\NC}{\bm{\mathsf{NC}}}
\newcommand{\TC}{\bm{\mathsf{TC}}}
\newcommand{\RNC}{\bm{\mathsf{RNC}}}
\newcommand{\ZPNC}{\bm{\mathsf{ZPNC}}}

\DeclareMathOperator{\Tr}{Tr}

\newcommand{\F}{\ensuremath{\mathbb{F}}}

\newcommand{\emphdef}[1]{{\sf {#1}}} % or \sf, or \textbf

\newcommand{\Bern}{\textsc{Bern}}

\newcommand{\Enc}{\textsc{Enc}}
\newcommand{\Dec}{\textsc{Dec}}

\newcommand{\oU}{\mathcal{U}}
\newcommand{\oA}{\mathcal{A}}
\newcommand{\oE}{\mathcal{E}}